\documentclass[letterpaper,11pt]{article}

\usepackage[margin=1in]{geometry}  
\usepackage{xspace,exscale,relsize}
\usepackage{fancybox,shadow}
\usepackage{graphicx}
\usepackage{color}
\usepackage{amsthm,amsfonts}
\usepackage{amssymb}
\usepackage{authblk}
\usepackage[title]{appendix}

\usepackage[noadjust]{cite}

\usepackage{amsmath}
	\makeatletter
	\let\over=\@@over \let\overwithdelims=\@@overwithdelims
	\let\atop=\@@atop \let\atopwithdelims=\@@atopwithdelims
  	\let\above=\@@above \let\abovewithdelims=\@@abovewithdelims
  	\makeatother
\usepackage{fixltx2e}

\usepackage{rotating}

\usepackage{ifpdf}

\usepackage{subfigure}
\usepackage{psfrag}

\usepackage{prettyref,enumerate}

\usepackage{tikz}
\usetikzlibrary{arrows}
\tikzstyle{int}=[draw, fill=blue!20, minimum size=2em]
\tikzstyle{dot}=[circle, draw, fill=blue!20, minimum size=2em]
\tikzstyle{init} = [pin edge={to-,thin,black}]

\usepackage[
            CJKbookmarks=true,
            bookmarksnumbered=true,
            bookmarksopen=true,
            colorlinks=true,
            citecolor=red,
            linkcolor=blue,
            anchorcolor=red,
            urlcolor=blue,
            pdfauthor={Polyanskiy and Wu}
            ]{hyperref}

\usepackage[all]{xy}

\usepackage{mathtools}

\newcommand{\matx}{\ensuremath{\mathcal{X}}}
\newcommand{\mata}{\ensuremath{\mathcal{A}}}

\newcommand{\matr}{\ensuremath{\mathcal{R}}}
\newcommand{\maty}{\ensuremath{\mathcal{Y}}}
\newcommand{\matn}{\ensuremath{\mathcal{N}}}

\newcommand{\mreals}{\ensuremath{\mathbb{R}}}

\newcommand{\supp}{\ensuremath{\mathrm{supp}}}

\ifx\eqref\undefined
	\newcommand{\eqref}[1]{~(\ref{#1})}
\fi
\ifx\mod\undefined
	\def\mod{\mathop{\rm mod}}
\fi

\def\exp{\mathop{\rm exp}}

\def\EE{\Expect}

\def\PP{\mathbb{P}}
\def\FF{\mathbb{F}}

\def\eqdef{\triangleq}

\def\gacap#1{{1\over 2} \log\left(1+{#1}\right)}

\newcommand{\etaKL}{\eta_{\rm KL}}

\newcommand{\etaTV}{\eta_{\rm TV}}

\newcommand{\reals}{\mathbb{R}}

\newcommand{\Expect}{\mathbb{E}}

\newcommand{\prob}[1]{\mathbb{P}\left[#1\right]}

\newcommand{\TV}{d_{\rm TV}}
\newcommand{\dbar}{\bar{d}}

\newcommand{\diff}{{\rm d}}
\newcommand{\eg}{e.g.\xspace}
\newcommand{\ie}{i.e.\xspace}

\newcommand{\pth}[1]{\left( #1 \right)}
\newcommand{\qth}[1]{\left[ #1 \right]}
\newcommand{\sth}[1]{\left\{ #1 \right\}}

\newcommand\indep{\protect\mathpalette{\protect\independenT}{\perp}}
\def\independenT#1#2{\mathrel{\rlap{$#1#2$}\mkern2mu{#1#2}}}
\newcommand{\Bern}{\text{Bern}}
\newcommand{\iprod}[2]{\left \langle #1, #2 \right\rangle}
\newcommand{\Iprod}[2]{\langle #1, #2 \rangle}
\newcommand{\indc}[1]{{\mathbf{1}_{\left\{{#1}\right\}}}}

\definecolor{myblue}{rgb}{.8, .8, 1}
\definecolor{mathblue}{rgb}{0.2472, 0.24, 0.6} 
\definecolor{mathred}{rgb}{0.6, 0.24, 0.442893}
\definecolor{mathyellow}{rgb}{0.6, 0.547014, 0.24}

\newcommand{\tp}{\tilde{p}}

\newcommand{\tx}{{\tilde{x}}}

\newcommand{\tA}{{\tilde{A}}}

\newcommand{\tU}{{\tilde{U}}}
\newcommand{\tV}{{\tilde{V}}}

\newcommand{\tX}{{\tilde{X}}}
\newcommand{\tY}{{\tilde{Y}}}

\newcommand{\sfb}{{\mathsf{b}}}

\newcommand{\calN}{{\mathcal{N}}}

\newcommand{\calR}{{\mathcal{R}}}

\newcommand{\calX}{{\mathcal{X}}}

\newcommand{\bard}{{\bar{d}}}

\newcommand{\diverge}{\to \infty}

\newrefformat{eq}{(\ref{#1})}
\newrefformat{thm}{Theorem~\ref{#1}}
\newrefformat{th}{Theorem~\ref{#1}}
\newrefformat{chap}{Chapter~\ref{#1}}
\newrefformat{sec}{Section~\ref{#1}}
\newrefformat{algo}{Algorithm~\ref{#1}}
\newrefformat{fig}{Fig.~\ref{#1}}
\newrefformat{tab}{Table~\ref{#1}}
\newrefformat{rmk}{Remark~\ref{#1}}
\newrefformat{clm}{Claim~\ref{#1}}
\newrefformat{def}{Definition~\ref{#1}}
\newrefformat{cor}{Corollary~\ref{#1}}
\newrefformat{lmm}{Lemma~\ref{#1}}
\newrefformat{prop}{Proposition~\ref{#1}}
\newrefformat{pr}{Proposition~\ref{#1}}
\newrefformat{app}{Appendix~\ref{#1}}
\newrefformat{apx}{Appendix~\ref{#1}}
\newrefformat{ex}{Example~\ref{#1}}
\newrefformat{exer}{Exercise~\ref{#1}}
\newrefformat{soln}{Solution~\ref{#1}}

\def\unifto{\mathop{{\mskip 3mu plus 2mu minus 1mu%
	\setbox0=\hbox{$\mathchar"3221$}%
	\raise.6ex\copy0\kern-\wd0%
	\lower0.5ex\hbox{$\mathchar"3221$}}\mskip 3mu plus 2mu minus 1mu}}

\ifx\lesssim\undefined
\def\simleq{{{\mskip 3mu plus 2mu minus 1mu%
	\setbox0=\hbox{$\mathchar"013C$}%
	\raise.2ex\copy0\kern-\wd0%
	\lower0.9ex\hbox{$\mathchar"0218$}}\mskip 3mu plus 2mu minus 1mu}}
\else
\def\simleq{\lesssim}
\fi

\ifx\gtrsim\undefined
\def\simgeq{{{\mskip 3mu plus 2mu minus 1mu%
	\setbox0=\hbox{$\mathchar"013E$}%
	\raise.2ex\copy0\kern-\wd0%
	\lower0.9ex\hbox{$\mathchar"0218$}}\mskip 3mu plus 2mu minus 1mu}}
\else
\def\simgeq{\gtrsim}
\fi

\newtheorem{theorem}{Theorem}
\newtheorem{lemma}[theorem]{Lemma}
\newtheorem{corollary}[theorem]{Corollary}
\newtheorem{coro}[theorem]{Corollary}
\newtheorem{proposition}[theorem]{Proposition}
\newtheorem{prop}[theorem]{Proposition}

\theoremstyle{definition}

\newtheorem{remark}{Remark}

\newif\ifmapx
{\catcode`/=0 \catcode`\\=12/gdef/mkillslash\#1{#1}}
\edef\jobnametmp{\expandafter\string\csname ic_apx\endcsname}
\edef\jobnameapx{\expandafter\mkillslash\jobnametmp}
\edef\jobnameexpand{\jobname}
\ifx\jobnameexpand\jobnameapx
\mapxtrue
\else
\mapxfalse
\fi

\long\def\apxonly#1{\ifmapx{\color{blue}#1}\fi}

\newcommand{\Lip}{\mathop{\mathrm{Lip}}}

\begin{document}
\ifpdf
\DeclareGraphicsExtensions{.pgf}
\graphicspath{{figures/}{plots/}}
\fi

\title{Wasserstein continuity of entropy and outer bounds for interference channels}

\author{Yury Polyanskiy and Yihong Wu\thanks{Y.P. is with the Department of EECS, MIT, Cambridge, MA, email: \url{yp@mit.edu}. Y.W. is with
the Department of ECE and the Coordinated Science Lab, University of Illinois at Urbana-Champaign, Urbana, IL, email: \url{yihongwu@illinois.edu}.}}


\maketitle

\begin{abstract}
It is shown that under suitable regularity conditions, differential entropy is $O(\sqrt{n})$-Lipschitz as a function of probability
distributions on $\reals^n$ with respect to the quadratic Wasserstein distance.  Under similar conditions, 
(discrete) Shannon entropy is shown to be $O(n)$-Lipschitz in distributions over the product space with respect to Ornstein's $\bar d$-distance (Wasserstein distance corresponding to the Hamming distance). These results together with Talagrand's and Marton's transportation-information
inequalities allow one to replace the unknown multi-user interference with its i.i.d.~approximations. As an application, a 
new outer bound for the two-user Gaussian interference channel is proved, which, in particular, settles the ``missing corner point'' problem of Costa (1985).  
\end{abstract}

\section{Introduction}
	\label{sec:intro}
	
	Let $X$ and $\tX$ be random vectors in $\reals^n$. We ask the following question: If the distributions of $X$ and $\tX$ are close in certain sense, can we guarantee that their differential entropies are close as well? For example, one can ask whether 
	\begin{equation}
	D(P_X\|P_\tX) = o(n) \xRightarrow{?} |h(X)-h(\tX)| = o(n).
	\label{eq:Dh}
\end{equation} 
One motivation comes from multi-user information theory, where frequently one user causes interference to the other and in proving the converse
one wants to replace the complicated non-i.i.d. interference by a simpler i.i.d. approximation. As a concrete example,
we consider the so-called ``missing corner point'' problem in the capacity region of the two-user
Gaussian interference channels (GIC) \cite{Costa85}. Perhaps due to the explosion in the number of interfering radio
devices, this problem has attracted renewed attention recently~\cite{costa2011noisebergs,Bustin-ISIT14,CR15,RC15}. For further information on capacity region
of GIC and especially the problem of corner points, we refer to a comprehensive account just published by Igal
Sason~\cite{IS13-GIC-corner}. 

Mathematically, the key question for settling ``missing corner point'' is the following: Given independent $n$-dimensional random vectors $X_1,X_2,G_2,Z$ with the latter two
being Gaussian, is it true that
\begin{equation}\label{eq:costa_req}
	D(P_{X_2+Z}\|P_{G_2 +Z}) = o(n) \xRightarrow{?} |h(X_1+X_2+Z)-h(X_1 + G_2 + Z)| = o(n).
\end{equation}

	To illustrate the nature of the problem, we first note that the answer to \prettyref{eq:Dh} is in fact negative as the counterexample of $X \sim \matn(0, 2 I_n)$ and $\tilde X\sim
{1\over2} \matn(0,I_n) + {1\over2}\matn(0,2I_n)$ demonstrates, in which case the divergence is $D(P_X\|P_\tX) \leq \log 2$ but the differential entropies differ by
$\Theta(n)$. Therefore even for very smooth densities the difference in entropies is not controlled by the divergence.  
The situation for discrete alphabets is very similar, in the sense that the gap of Shannon entropies cannot be bounded by divergence in general (with essentially the same counterexample as that in the continuous case: $X$ and $\tX$ being uniform on one and
two Hamming spheres respectively). 	

The rationale of the above discussion is two-fold: a) Certain regularity conditions of the distributions must be imposed; b) Distances other than KL divergence might be more suited for bounding the entropy difference.
Correspondingly, the main contribution of this paper is the following: Under suitable regularity conditions, the difference in entropy (in both continuous and discrete cases) can in fact be bounded by the \emph{Wasserstein distance}, a notion originating from optimal transportation theory which turns out to be the main tool of this paper. 

We start with the definition of the Wasserstein distance on the Euclidean space. Given probability measures $P,Q$ on $\reals^n$, define their $p$-Wasserstein distance ($p\geq 1$) as
	\begin{equation}
	W_p(P,Q) \triangleq \inf (\Expect[\|X-Y\|^p])^{1/p},
	\label{eq:wp}
\end{equation}
where $\|\cdot\|$ denotes the Euclidean distance and 
the infimum is taken over all couplings of $P$ and $Q$, \ie, joint distributions $P_{XY}$ whose marginals satisfy $P_X=P$ and $P_Y=Q$. 
The following dual representation of the $W_1$ distance is useful: 
\begin{equation}
	W_1(P,Q) = \sup_{\Lip(f) \leq 1}  \int f dP - \int f dQ.
	\label{eq:w1dual}
\end{equation}

Similar to \prettyref{eq:Dh}, it is easy to see that
in order to control $ |h(X) - h(\tX)|$ by means of $W_2(P_{X}, P_{\tX}) $,
one necessarily needs to assume some regularity properties of $P_{X}$ and $P_{\tX}$; otherwise, choosing one to
be a fine quantization of the other creates infinite gap between differential entropies, while keeping the $W_2$ distance arbitrarily small.
Our main result in \prettyref{sec:wass} shows that under moment constraints and certain conditions on the densities (which are in particular satisfied by convolutions with Gaussians), various information measures such as differential entropy and mutual information on $\reals^n$ are in fact $\sqrt{n}$-Lipschitz continuous with respect to the $W_2$-distance. These results have natural counterparts in the discrete case where the Euclidean distance is replaced by Hamming distance (\prettyref{sec:discrete}).

Furthermore, \emph{transportation-information inequalities}, such as those due to Marton \cite{Marton86} and Talagrand
\cite{Talagrand96}, allow us to bound the Wasserstein distance by the KL divergence (see, \eg, \cite{raginsky2013concentration} for a review). For example, Talagrand's inequality states that if $Q=\calN(0,\Sigma)$, then
\begin{equation}
	W_2^2(P,Q) \le  \frac{2 \sigma_{\max}(\Sigma)}{\log e}
	D(P \| Q) \,,
	\label{eq:talagrand}
\end{equation}
where $\sigma_{\max}(\Sigma)$ denotes the maximal singular value of $\Sigma$. Invoking~\eqref{eq:talagrand} in conjunction with
the Wasserstein continuity of the differential entropy, we establish~\eqref{eq:costa_req} and prove a new outer bound for the capacity region of the
two-user GIC, finally settling the missing corner point in \cite{Costa85}. See \prettyref{sec:gic} for
details.

One interesting by-product is an estimate  that goes in the reverse direction 
of~\eqref{eq:talagrand}. Namely, under regularity conditions on $P$ and $Q$ we have\footnote{For positive $a,b$, denote $a \simleq b$ if $a/b$ is at most some universal constant.}
\begin{equation}\label{eq:dgg}
	D(P\|Q) \simleq \sqrt{\int_{\mreals^n} \|x\|^2 (dP + dQ)} \cdot W_2(P,Q)\,
\end{equation}
See~\prettyref{prop:ppr} and \prettyref{cor:w2lip} in the next section. We want to emphasize that there are a number of
estimates of the form $D(P_{X+Z}\|P_{\tilde X+Z}) \simleq W_2^2(P_X,P_{\tilde X})$ where $\tilde X, X$ are independent of a standard Gaussian vector $Z$, cf.~\cite[Chapter
9, Remark 9.4]{villani.topics}. 
The key difference of these estimates from~\eqref{eq:dgg} is that the $W_2$ distance is measured \textit{after}
convolving with $P_Z$.

\paragraph{Notations}	
Throughout this paper $\log$ is with respect to an arbitrary base, which also specifies the units of differential entropy $h(\cdot)$, Shannon entropy
$H(\cdot)$, mutual information $I(\cdot; \cdot)$ and divergence $D(\cdot\|\cdot)$. The natural logarithm is denoted by
$\ln$. The norm of $x\in\reals^n$ is denoted by 
$\|x\| \eqdef (\sum_{j=1}^n x_j^2)^{1/2}$.
For random variables $X$ and $Y$, let $X\indep Y$ denote their independence.

\section{Wasserstein-continuity of information quantities}
\label{sec:wass}

We say that a probability density function $p$ on $\mreals^n$ is $(c_1,c_2)$-regular if $c_1>0, c_2\ge0$ and
$$ \|\nabla \log p(x)\| \le c_1\|x\| + c_2, \qquad \forall x\in\mreals^n\,.$$
Notice that in particular, regular density is never zero and furthermore
$$ |\log p(x) - \log p(0)| \le {c_1\over2} \|x\|^2 + c_2 \|x\|$$
Therefore, if $X$ has a regular density and finite second moment then
$$ |h(X)| \le |\log P_X(0)| + c_2 \EE[\|X\|] + {c_1\over2} \EE[\|X\|^2] < \infty\,. $$

\begin{prop}
\label{prop:ppr}
 Let $U$ and $V$ be random vectors with finite second moments. If $V$ has a $(c_1,c_2)$-regular density $p_V$, then there exists a coupling $P_{UV}$, such that
\begin{equation}\label{eq:ppr1}
	\EE\left[\left|\log {p_{V}(V)\over p_{V}(U)}\right|\right] \le  \Delta \,,
\end{equation}
where 
$$ \Delta = \Big(\frac{c_1}{2} \sqrt{\EE[\|U\|^2]} + \frac{c_1}{2}\sqrt{\EE[\|V\|^2]}+c_2\Big) W_2(P_{U}, P_{V})\,.$$
Consequently,
\begin{equation}\label{eq:ppr2}
	h(U) - h(V) \le \Delta.
\end{equation}
If both $U$ and $V$ are $(c_1,c_2)$-regular, then
\begin{align} 
	|h(U) - h(V)| &\le \Delta,\label{eq:ppr3}\\
	D(P_{U}\|P_{V}) + D(P_{V}\|P_{U}) &\le 2 \Delta. \label{eq:ppr4}
\end{align}
\end{prop}
\begin{proof} 
First notice:
\begin{align} |\log p_{V}(v) - \log p_{V}(u)| &= \left|\int_0^1 dt \iprod{\nabla \log p_{V}(tv +(1-t)u)}{u-v}\right|\\
		&\le \int_0^1 dt (c_2+c_1t\|v\| + c_1(1-t)\|u\|) \|u-v\| \label{eq:ppr1a0}\\
		&  = (c_2+ c_1 \|v\|/2 + c_1 \|u\| /2)\|u-v\|, \label{eq:ppr1a}
\end{align}
where \prettyref{eq:ppr1a0} follows from Cauchy-Schwartz inequality and the $(c_1,c_2)$-regularity of $p_V$.
Taking expectation of~\eqref{eq:ppr1a} with respect to $(u,v)$ distributed according to the optimal
$W_2$-coupling of $P_U$ and $P_V$ and then applying Cauchy-Schwartz and triangle inequality for $L_2$-norm, we obtain \prettyref{eq:ppr1}. 

To show \eqref{eq:ppr2} notice that by finiteness of second moment
$h(U)<\infty$. If $h(U)=-\infty$ then there is nothing to prove.  So assume otherwise, then in
identity
\begin{equation}\label{eq:ppr2a}
	h(U) - h(V) + D(P_{U}\|P_{V}) = \EE\left[\log {p_{V}(V)\over p_{V}(U)}\right] 
\end{equation}
all terms are finite and hence~\eqref{eq:ppr2} follows. Clearly,~\eqref{eq:ppr2} implies~\eqref{eq:ppr3} (when applied
with $U$ and $V$ interchanged).

Finally, for~\eqref{eq:ppr4} just add the identity~\eqref{eq:ppr2a} to itself with $U$ and $V$ interchanged to obtain 
$$ D(P_{U}\|P_{V}) + D(P_{V}\|P_{U}) = 
		\EE\left[\log {p_{V}(V)\over p_{V}(U)}\right]  + \EE\left[\log {p_{U}(U)\over p_{U}(V)}\right]$$
and estimate both terms via~\eqref{eq:ppr1}.
\end{proof}

%

The key question now is what densities are regular. It turns out that convolution with sufficiently smooth density, such as Gaussians,
produces a regular density. 
\begin{prop}\label{prop:pqr} Let $V=B + Z$ where $B\indep Z\sim\matn(0, \sigma^2 I_n)$ and $\EE[\|B\|] < \infty$. Then the density
of $V$ is $(c_1,c_2)$-regular with $c_1 = {3\log e\over\sigma^2}$ and $c_2={4 \log e\over \sigma^2} \EE[\|B\|]$. 
\end{prop}
\begin{proof} First notice that whenever density $p_Z$ of $Z$ is differentiable and non-vanishing, we have:
\begin{equation}\label{eq:pqrX}
	\nabla \log p_{V}(v) =  \frac{\EE[\nabla p_{Z}(v-B)]}{p_V(v)} =  \EE[\nabla \log p_{Z}(v-B)|V=v]\,,
\end{equation}
where $p_V(v)=\EE[p_{Z}(v-B)]$ is the density of $V$.
For $Z\sim\matn(0, \sigma^2 I_n)$, we have
$$ \nabla \log p_{Z}(v-B) = {\log e\over \sigma^2} (B-v). $$
So the proof is completed by showing
\begin{equation}\label{eq:pqr0}
	\EE[\|B-v\| \,|\,V=v] \le 3\|v\| + 4\EE[\|B\|]\,.
\end{equation}
For this, we mirror the proof in~\cite[Lemma 4]{mmse.functional.IT}. Indeed, we have
\begin{align} \EE[\|B-v\| |V=v] &= \EE\qth{\|B-v\| {p_{Z}(B-v)\over p_{V}(v)}} \\
		   &\le 2\EE[\|B-v\|1\{a(B,v)\le 2\}] + \EE[\|B-v\|a(B,v) 1\{a(B,v)>2\}]\,, \label{eq:pqr1}
\end{align}		   
		   where we denoted
		   $$ a(B,v) \eqdef {p_{Z}(B-v)\over p_{V}(v)}\,.$$
Next, notice that
$$ \{a(B,v) > 2\} = \{\|B-v\|^2 \le -2\sigma^2 \ln((2\pi\sigma^2)^{n/2} 2 p_{V}(v))\}\,.$$
Thus since $\EE[p_{Z}(B-v)] = p_{V}(v)$ we have an upper bound for the second term in~\eqref{eq:pqr1} as follows
\begin{equation}\label{eq:pqr2}
	\EE[\|B-v\|a(B,v) 1\{a(B,v)>2\}] \le \sqrt{2}\sigma \sqrt{\ln^+{1\over(2\pi\sigma^2)^{{n\over2}} 2 p_{V}(v)}}\,, 
\end{equation}
where $\ln^+ x \triangleq \max\{0,\ln x\}$.
From Markov inequality we have $ \PP[\|B\| \le 2 \EE[\|B\|]] \ge 1/2$
and therefore
$$ p_{V}(v) \ge {1\over 2 (2\pi \sigma^2)^{n\over2}} e^{-{(\|v\| + 2 \EE[\|B\|])^2\over 2\sigma^2}}\,. $$
Using this estimate in~\eqref{eq:pqr2} we get
\begin{equation}\label{eq:pqr3}
	\EE[\|B-v\|a(B,v) 1\{a(B,v)>2\}] \le \|v\| + 2 \EE[\|B\|]\,. 
\end{equation}
Upper-bounding the first term in~\eqref{eq:pqr1} by $2\EE[\|B\|] + 2\|v\|$ we finish the proof of~\eqref{eq:pqr0}.
\end{proof}

Another useful criterion for regularity is the following:
\begin{prop} If $W$ has $(c_1,c_2)$-regular density and $B\indep W$ satisfies
	\begin{equation}\label{eq:power_as}
		\|B\| \le \sqrt{nP} \qquad \mbox{a.s.} 
\end{equation}	
	then $V = B+W$ has $(c_1,c_2 + c_1\sqrt{nP})$-regular density.
\end{prop}
\begin{proof} Apply~\eqref{eq:pqrX} and the estimate:
$$ \EE[ \|\nabla \log p_{W}(v-B)\| \, |\, V=v ] \le c_1(\|v\| + \sqrt{nP}) + c_2. \qedhere$$
\end{proof}

As a consequence of regularity, we show that when smoothed by Gaussian noise, mutual information, differential entropy and divergence are Lipschitz with respect to the $W_2$-distance under average power constraints:

\begin{corollary} 
\label{cor:w2lip}
Assume that $X,\tX \indep Z$, with $\EE[\|X\|^2], \EE[\|\tX\|^2]  \leq {nP} $ and $Z \sim \calN(0,\sigma^2 I_n)$. Then
\begin{align} 
|I(X; X + Z) - I(\tilde X; \tilde X + Z)| =	|h(X + Z) - h(\tilde X + Z)| &\leq \Delta,
		\label{eq:tpr1}\\
 D(P_{\tilde X + Z} \| P_{X + Z}) +   D(P_{X + Z} \| P_{\tilde X + Z}) &\leq 2 \Delta,
\end{align}
where $\Delta = \frac{\log e}{\sigma^2} (3\sqrt{n(\sigma^2 +P)} + 4 \sqrt{nP}) W_2(P_{X + Z}, P_{\tilde X + Z})$.
\end{corollary}
\begin{proof} Since $\Expect[\|X\|] \leq \sqrt{nP}$, by Proposition~\ref{prop:pqr}, the densities of $X + Z$ and $\tX + Z$ are both $(\frac{3 \log e}{\sigma^2}, \frac{4\sqrt{n P} \log e}{\sigma^2} )$-regular.
The desired statement then follows from applying~\eqref{eq:ppr3}-\eqref{eq:ppr4} to $V=X+Z$ and $U=\tX+Z$.
\end{proof}

\begin{remark} The Lipschitz constant $\sqrt{n}$ is order-optimal as the example of Gaussian $X$ and $\tilde X$ with
different variances (one of them could be zero) demonstrates. The linear dependence on $W_2$ is also optimal. To see this, consider $X\sim \calN(0,1)$ and $\tX\sim \calN(0,1+t)$ in one dimension. Then $|h(X+Z)-h(\tX+Z)| = 1/2 \log(1+t/2) = \Theta(t)$ and $W_2^2(X+Z,\tX+Z) = (\sqrt{2+t}-\sqrt{2})^2 = \Theta(t^2)$, as $t \to 0$.
\end{remark}

In fact, to get the best constants for applications to interference channels it is best to forgo the notion of regular density and
deal directly with~\eqref{eq:pqrX}. Indeed, when the inputs has bounded norms, the next result gives a sharpened version of what can be obtained by combining \prettyref{prop:ppr} with \ref{prop:pqr}.
\begin{prop}\label{prop:best} Let $B$ satisfying \prettyref{eq:power_as} and $G \sim \matn(0, \sigma_G^2 I_n)$ be independent. Let $V = B + G$.  Then for any $U$, 
\begin{align}
h(U) - h(V) 
		&\le {\log e \over 2\sigma_G^2}\left(\EE[\|U\|^2] - \EE[\|V\|^2] + 2 \sqrt{nP} W_1(P_{U}, P_{V})\right)\,.\label{eq:bd_best}
\end{align}		
\end{prop}
\begin{proof} Plugging Gaussian density $p_G(z)= \frac{1}{\sqrt{2\pi}\sigma_G} e^{-z^2/(2\sigma_G^2)}$ into~\eqref{eq:pqrX} we get
\begin{equation}\label{eq:pqrXX}
		\nabla \log p_{V}(v) = {\log e\over\sigma_G^2}(\hat B(v) - v)\,,
\end{equation}	
	where $\hat B(v) \triangleq \EE[B|V=v] = \frac{\Expect[B p_G(v - B)]}{\Expect[p_G(v - B)]}$ satisfies
		$$ \|\hat B(v)\| \le \sqrt{nP},$$
since $\|B\| \leq \sqrt{nP}$ almost surely.
Next we use
\begin{align} \log {p_{V}(v)\over p_{V}(u)} &= \int_0^1 dt \iprod{\nabla \log p_{V}(tv +(1-t)u)}{v-u} \label{eq:s1}\\
				&={\log e\over\sigma_G^2} \int_0^1 dt \Iprod{\hat B(tv + (1-t)u)}{ v-u} - {\log e\over2\sigma_G^2} (\|v\|^2-\|u\|^2)\\
				&\le {\sqrt{nP}\log e\over\sigma_G^2} \|v-u\| - {\log e\over2\sigma_G^2} (\|v\|^2-\|u\|^2)\,.
\end{align}			
Taking expectation of the last equation under the $W_1$-optimal coupling and in view of \prettyref{eq:ppr2a}, we obtain \prettyref{eq:bd_best}.
\end{proof}

To get slightly better constants in one-sided version of~\eqref{eq:tpr1} we apply Proposition~\ref{prop:best}:
\begin{corollary}\label{cor:best} 
Let $A,B,G,Z$ be independent, with $G\sim \calN(0,\sigma_G^2 I_n)$, $Z \sim \calN(0, \sigma_Z^2 I_n)$ and $B$ satisfying~\eqref{eq:power_as}. Then for every $c\in[0,1]$ we have:
\begin{align}
 & ~ h(B + A + Z) - h(B + G + Z)	\nonumber \\
\leq & ~ 	{\log e\over2(\sigma_G^2 + \sigma_Z^2)} \left(\EE[\|A\|^2] + 2 \iprod{\Expect[A]}{\Expect[B]}- \EE[\|G\|^2]\right) + 
		{\sqrt{2 n P (\sigma_G^2 + c^2\sigma_Z^2)\log e}\over \sigma_G^2 + \sigma_Z^2 } \sqrt{ D(P_{A + c Z} \| P_{G + c
		Z})}
\end{align}
\end{corollary}
\begin{proof}
First, notice that by definition Wasserstein distance is non-increasing under convolutions, \ie, $W_2(P_1*Q,P_2*Q) \leq
W_2(P_1,P_2)$. Since $c\le 1$ and Gaussian distribution is stable, we have
$$ W_2(P_{B + A + Z}, P_{B + G + Z}) \leq W_2(P_{A + Z}, P_{G + Z}) \le W_2(P_{A + c Z}, P_{G + c Z}), $$
which, in turn, can be bounded via Talagrand's inequality \prettyref{eq:talagrand} by
$$ W_2(P_{A + c Z}, P_{G + c Z}) \le \sqrt{{2(\sigma_G^2 + c^2\sigma_Z^2)\over \log e} D(P_{A + cZ}\|P_{G+cZ})}\,. $$
From here we apply Proposition~\ref{prop:best} with $G$ replaced by $G+Z$ (and $\sigma_G^2$ by $\sigma_Z^2+\sigma_G^2$).
\end{proof}
\apxonly{Note: If there is extra information allowing one to conclude that
$$ D(A + Z \| G + Z) \ll D(A + cZ \| G + cZ) $$
then it could be used to improve the estimate 
$$ W_2(A + Z, G + Z) $$
by applying Talagrand wrt $G+Z$. However, the standard log-Sobolev estimate~\eqref{eq:ptt1} is not sufficient:
applying it yields exactly the same bound as using Talagrand on reduced noise $G+cZ$.}

\apxonly{
Alternative bound, which has the correct logarithmic dependence but unfortunately does no vanish when $P_A=P_G$:

\begin{proposition}
	Let $Y=A+Z_1$, $Y_G = G + Z_2$ with $Z_1$ and $Z_2$ being standard Gaussian. Then
\[
|h(Y) - h(Y_G)| \leq n/2 \log (W_2^2(A,G)/n + 2).
\]
\end{proposition}
\begin{proof}
Couple $A$ and $G$ according to the optimal $W_2$ coupling and take $Z_1 \indep Z_2$. Then notice:
$$
h(Y) - h(Y_G) = h(Y|Y_G) - h(Y_G|Y).
$$
Using $h(U|V) \leq n/2 \log (2\pi e \EE[\|U-V\|^2]/n)$, we get
$h(A+Z_1 | G+Z_2) \leq n/2 \log (2 \pi e (W_2^2(A,G)/n + 2))$, where we used $\EE[\|A+Z_1 - (G+Z_2)\|^2] = \EE[\|A - G\|^2] + 2n $. On the other hand, we have
$h(G+Z_2 | A+Z_1) \geq h(Z_2 | A+Z_1) = h(Z_2) = n/2 \log 2\pi e$.
\end{proof}
}

\apxonly{
\subsection{High noise to small noise}

Here is how to tradeoff divergence at different noise values.
\begin{prop} Let $X$ be arbitrary $n$-dimensional distribution and let
\begin{align} f(t) &\eqdef D(e^{-t} X + \sqrt{1-e^{-2t}}Z\|Z)
\end{align}   
For all $t>t_0$ we have
\begin{align}\label{eq:ptt1}
	f(t) &\le f(t_0) e^{-2(t-t_0)}\\
\end{align}
For all $t<t_0$ we have
\begin{equation}\label{eq:ptt2}
	f(t) \le f(t_0) + 4(t_0-t_1)\left[ e^{-2t} \EE[\|X\|^2] + n {e^{-4t}\over 1-e^{-2t}} \right]\log e\,.
\end{equation}
\end{prop}
\begin{proof} Estimate~\eqref{eq:ptt1} are the standard log-Sobolev inequality estimate~\cite[Section
9.2]{villani.topics}. Note that Villani also proposes the following:
\begin{align} w(t) &\eqdef W_2(e^{-t} X + \sqrt{1-e^{-2t}}Z, Z)\,,
	w(t) &\le w(t_0) e^{-(t-t_0)}\label{eq:ptt1a}\,.
\end{align}	
But this is a triviality following from scaling and the obvious
$$ W_2(X_1 + Z', X_2 + Z') \le W_2(X_1, X_2)\,.$$

For~\eqref{eq:ptt2}, notice that the function $f(t)$ is decreasing and convex (standard Bakry-Emery business). Thus it
is sufficient to show that 
\begin{equation}\label{eq:ptt3}
	f'(t) \ge - 4\left[ e^{-2t} \EE[\|X\|^2] + n {e^{-4t}\over 1-e^{-2t}} \right]\,.
\end{equation}
Straightforward differentiation yields:
\begin{align} f'(t) &= -4\EE\left[\left|\nabla \log {dP_{Y_t}\over dP_{Z}}(Y_t)\right|^2\right]\\
	&= -4 \EE[|Y_t + \nabla \log p_{Y_t}(Y_t)|^2]\,,
\end{align}	
where $Y_t=e^{-t} X + \sqrt{1-e^{-2t}}Z$ and $p_{Y_t}$ is its pdf. Then, using~\eqref{eq:pqrX} and using
$$ \EE\left[ \left|\EE[e^{-t} X | Y_t] \right|^2 \right] \le e^{-2t} \EE[\|X\|^2] $$
we get~\eqref{eq:ptt3}.
\end{proof}

Note that the bound in~\eqref{eq:ptt2} does not show $f(t_0)\to0$ when $f(t)\to0$. In fact such a bound is not possible
(at least in the limit of large $n$). Here is an exact statement:

\begin{proposition} There exists random variables $X_n\in \mreals^n$ such that
$$ {1\over n} D(X_n + Z\| \matn(0, (1+P)I_n)) \to 0 $$
but for any $t_0 < 1$ there is $\delta_0>0$ s.t.
$$ {1\over n} D(X_n + \sqrt{t_0} Z\| \matn(0, (t_0+P)I_n)) \ge \delta_0 + o(1)\,. $$
\end{proposition}
\begin{proof} Let $X_n$ be distribution induced by the power-$P$ codebook of rate $R={1\over 2}\log(1+P)+o(1)$ and decodable with vanishing probability of error
from $X_n+Z$. Then, we have
\begin{align} D(X_n + \sqrt{t_0} Z\| \matn(0, (t_0+P)I_n)) &= {n\over 2} \log(1+P/t_0)-I(X_n; X_n + \sqrt{t_0} Z) 
			\\&\ge
			{n\over 2} \log(1+P/t_0) - H(X_n) \\
			&= {n\over2} \log {1+P/t_0\over 1+P} + o(n) 
\end{align}			
\end{proof}
}

\section{Applications to Gaussian interference channels}
\label{sec:gic}

\subsection{New outer bound}

Consider the two-user Gaussian interference channel (GIC):
\begin{equation}
\begin{aligned} Y_1 &= X_1 + bX_2 + Z_1  \\
   Y_2 &= aX_1+ X_2 + Z_2 \,,
\end{aligned}	
	\label{eq:GIC}
\end{equation}   
with $a,b\geq 0$, $Z_i\sim\matn(0, I_n)$ and a power constraint on the $n$-letter codebooks: either
\begin{equation}
	\|X_1\| \le \sqrt{nP_1}, \quad \|X_2\| \le \sqrt{nP_2}\, \quad \text{a.s.}
	\label{eq:power}
\end{equation}
or
\begin{equation}
	\EE[\|X_1\|^2] \le {nP_1}, \quad \Expect[\|X_2\|^2] \le nP_2.
	\label{eq:power2}
\end{equation}	
Denote by $\calR(a,b)$ the capacity region of the GIC \prettyref{eq:GIC}. As an application of the results developed in \prettyref{sec:wass}, we prove an outer bound for the capacity region.

\apxonly{
\begin{theorem}
Let $a \leq 1$.	The capacity region of the GIC satisfies:
\begin{equation}
R_1 \le C_1' + g(C_2-R_2)
	\label{eq:corner}
\end{equation}
where $C_2 = {1\over 2} \log (1+P_2)$, $C_1' = {1\over 2} \log(1+\frac{a^2 P_1}{1+P_2})$ and 
\begin{equation}\label{eq:gdef}
g(s) = \frac{1}{2} \log \pth{1 + \frac{1-a^2+ P_2}{a^2 P_2} (\exp(2 \delta ) - 1)  } 
\end{equation}
and $\delta = s + a\sqrt{\frac{2 P_1 s \log e }{1+P_2}}$.
	\label{thm:corner}
\end{theorem}

\begin{proof}
First of all, for $i=1,2$ define 
$$ R_i \eqdef {1\over n} I(X_i; Y_i)\,.$$
(By Fano's inequality there is no difference asymptotically between this definition of rate and the operational one.)

Let $G_2 \sim \matn(0, P_{2} I_n)$. Then
\[
n R_2 = I(X_{2}; Y_{2}) \leq I(X_{2}; Y_{2}|X_1) = I(X_{2}; X_{2}+Z_2) \leq n C_2 - D( X_{2}+Z_2 \| G_{2}+Z_2 ),
\]
that is,
\begin{equation}
	D( X_{2}+Z_2 \| G_{2}+Z_2) \leq n (C_2-R_2).
	\label{eq:R2gap}
\end{equation}
Furthermore,
\begin{align} 
nR_2 = I(X_{2}; Y_{2}) &=h(a X_{1} + X_{2} + Z_{2}) - h(a X_{1}  + G_{2} + Z_{2}) \label{eq:t1}\\
		& {} + h(a X_{1}  + G_{2} + Z_{2}) - h(a X_{1} + Z_{2})\,. \label{eq:t2}
\end{align}
The first term \prettyref{eq:t1} can be bounded by applying \prettyref{cor:best} to \prettyref{eq:R2gap} with $B = a X_1$, $A = X_2$ and $G= G_2$ and
$c=1$:
\begin{equation}
	h(a X_{1} + X_{2} + Z_{2}) - h(a X_{1}  + G_{2} + Z_{2}) \leq 
		n \sqrt{\frac{2 a^2 P_1 (C_2-R_2)\log e }{(1+P_2)}} \triangleq \alpha n.
	\label{eq:t3}
\end{equation}
To relate the second term \prettyref{eq:t2} to $R_1$, define $f(t) \triangleq \exp(\frac{2}{n} h(P_{X_1} * \calN(0, tI_n)))$.
Then
\begin{equation}
R_2 \leq \alpha + \frac{1}{2} \log f\pth{\frac{1+P_2}{a^2}} - \frac{1}{2} \log f\pth{\frac{1}{a^2}}	
	\label{eq:R2f}
\end{equation}
On the other hand, 
\begin{equation}
	n R_1 = I(X_{1}; Y_{1}) = h(X_1+Z_1) - h(Z_1) =   \frac{n}{2} \log f\pth{1} - \frac{n}{2} \log(2 \pi e)
	\label{eq:R1}
\end{equation}
By the concavity of the entropy power \cite{Costa85}, $f$ is concave on $\reals_+$. Since $\bar \lambda  + \lambda \frac{1+P_2}{a^2} = \frac{1}{a^2}$ for $\lambda = \frac{1-a^2}{1-a^2+P_2}$, we have
\begin{equation}
\bar \lambda f(1) + \lambda f\pth{\frac{1+P_2}{a^2}}  \leq f\pth{\frac{1}{a^2}}.	
	\label{eq:epic}
\end{equation}
Since $\Expect[|X_1|^2] \leq n P_1$, we have
\begin{equation}
	f(t) \leq 2 \pi e (P_1+t).
	\label{eq:ft}
\end{equation}
Combining \prettyref{eq:R2f} -- \prettyref{eq:ft}, we obtain \prettyref{eq:corner}.
\end{proof}

\textbf{TODO:} Add remark that the proof also shows a statement like 
$$ D(X_1 + \sqrt{1+P_2\over a^2} Z \| G_1 + \sqrt{1+P_2\over a^2} Z) \le $$

\begin{remark} A simpler bound:
$$ R_1 \le C_1' + \tilde g(C_2-R_2) $$
where 
\[
\tilde g(s) = \frac{1-a^2+ P_2}{a^2 P_2}\left(s + a \sqrt{\frac{2  P_1 s }{(1+P_2)\log e}} \right)
\]
is obtained by applying $\log(1+cx)\le c \log x$ in~\eqref{eq:gdef}. Interestingly, this bound corresponds to invoking
concavity of
\begin{equation}\label{eq:iofgamma}
	\gamma \mapsto I(X_1; \sqrt{\gamma} X_1 + Z) 
\end{equation}
in place of Costa's EPI in~\eqref{eq:epic}. Costa's EPI implies concavity of~\eqref{eq:iofgamma} and thus results in a
stronger bound.
\end{remark}
}

\begin{theorem} 
	\label{thm:corner2}
	Let $0 < a\le1$. 
	Let $C_2 = {1\over 2} \log (1+P_2)$ and $\tilde C_2 = {1\over 2} \log(1+\frac{P_2}{1+a^2 P_1})$.
	Assume the almost sure power constraint \prettyref{eq:power}.
	Then for any $b \geq 0$ and $\tilde C_2 \leq R_2 \leq C_2$, any rate pair $(R_1,R_2) \in \calR(a,b)$ satisfies
\begin{equation}
R_1 \le {1\over 2} \log \min\left\{ A - {1\over a^2} + 1, A {(1+P_2) (1-(1-a^2)\exp(-2\delta))-a^2\over P_2}\right\}
	\label{eq:corner2}
\end{equation}
where 
\begin{align}A &= (P_1 + a^{-2}(1+P_2))\exp(-2R_2), \label{eq:A}\\ 
 \delta &= C_2-R_2 + a\sqrt{\frac{2 P_1 (C_2-R_2 ) \log e }{1+P_2}}.  \label{eq:delta}
\end{align} 
Assume the average power constraint \prettyref{eq:power2}. Then \prettyref{eq:corner2} holds with $\delta$ replaced by
\begin{align}
 \delta' = C_2-R_2 + \sqrt{\frac{2  (C_2-R_2)\log e }{1+P_2}}  (3 \sqrt{1+a^2 P_1+P_2} + 4 a \sqrt{P_1}) .
 \label{eq:delta-avg}
\end{align}
Consequently, in both cases, $R_2 \geq C_2 - \epsilon$ implies that $R_1 \leq \frac{1}{2} \log(1+\frac{a^2 P_1}{1+P_2}) - \epsilon'$ where $\epsilon' = O(\sqrt{\epsilon})$ as $\epsilon \to 0$.
\end{theorem}



\begin{proof} 
Without loss of generality, assume that all random variables have zero mean.
First of all, setting $b=0$ (which is equivalent to granting the first user access to $X_2$) will not shrink the capacity region of the interference channel \prettyref{eq:GIC}. Therefore to prove the desired outer bound it suffices to focus on the following Z-interference channel henceforth:
\begin{equation}
\begin{aligned} Y_1 &= X_1 + Z_1 \\
   Y_2 &= aX_1+ X_2 + Z_2 \,.
\end{aligned}
	\label{eq:ZGIC}
\end{equation}   
Let $(X_1,X_2)$ be $n$-dimensional random variables corresponding to the encoder output of the first and second user, which are uniformly distributed on the respective codebook. 
For $i=1,2$ define 
$$ R_i \eqdef {1\over n} I(X_i; Y_i)\,.$$
By Fano's inequality there is no difference asymptotically between this definition of rate and the operational one.
Define the entropy-power function of the $X_1$-codebook:
$$ N_1(t) \eqdef \exp\left\{{2\over n}h(X_1 + \sqrt{t} Z)\right\}\,, \qquad Z \sim \matn(0,I_n)\,.$$
We know the following general properties of $N_1(t)$:
\begin{itemize}
\item $N_1$ is monotonically increasing.
\item $N_1(0)=0$ (since $X_1$ is uniform over the codebook).
\item $N_1'(t) \ge 2\pi e$ (since $N_1(t+\delta) \geq N_1(t) + 2\pi e \delta$ by entropy power inequality).
\item $N_1(t)$ is concave (Costa's entropy power inequality~\cite{Costa85}).
\item $N_1(t) \le 2\pi e(P_1 +t)$ (Gaussian maximizes differential entropy).
\end{itemize}
We can then express $R_1$ in terms of the entropy power function as
\begin{equation}
	R_1 = {1\over 2} \log {N_1(1)\over 2\pi e}\,.
	\label{eq:R1N1}
\end{equation}
It remains to upper bound $N_1(1)$. Note that
$$ nR_2 = I(X_2; Y_2) = h(X_2 + aX_1 + Z) - h(a X_1 + Z) \le {n\over 2} \log 2\pi e(1+P_2+a^2 P_1) - h(a X_1 +
Z)\,,$$
and therefore
\begin{equation}\label{eq:ct3}
	N_1\left({1\over a^2}\right) \le 2\pi e A\,,
\end{equation}
where $A$ is defined in \prettyref{eq:A}. This in conjunction with the slope property $N_1'(t)\ge 2\pi e$ yields
\begin{equation}\label{eq:ct4}
	N_1(1) \le N_1\left({1\over a^2}\right) - 2\pi e (a^{-2} - 1) \leq  2\pi e (A - a^{-2}+1) \,,
\end{equation}
which, in view of \prettyref{eq:R1N1}, yields the first part of the bound \prettyref{eq:corner2}.

To obtain the second bound, let $G_2 \sim \matn(0, P_{2} I_n)$. Using $\Expect[\|X_2\|^2] \leq n P_2$ and $X_1 \indep X_2$, we obtain
\begin{align*}
n R_2 
= & ~ I(X_{2}; Y_{2}) \leq I(X_{2}; Y_{2}|X_1) = I(X_{2}; X_{2}+Z_2)	\nonumber \\
= & ~ n C_2 - h(G_{2}+Z_2) + h(X_{2}+Z_2) \leq n C_2 - D( P_{X_{2}+Z_2} \| P_{G_{2}+Z_2} ),\
\end{align*}
that is,
\begin{equation}
	D( P_{X_{2}+Z_2} \| P_{G_{2}+Z_2}) \leq h(G_{2}+Z_2) - h( X_{2}+Z_2) \leq n (C_2-R_2).
	\label{eq:R2gap}
\end{equation}
Furthermore,
\begin{align} 
nR_2  = I(X_{2}; Y_{2}) &=h(a X_{1} + X_{2} + Z_{2}) - h(a X_{1}  + G_{2} + Z_{2}) \label{eq:t1}\\
		& {} + h(a X_{1}  + G_{2} + Z_{2}) - h(a X_{1} + Z_{2})\,. \label{eq:t2}
\end{align}
Note that the second term \prettyref{eq:t2} is precisely $\frac{n}{2} \log \frac{N_1(\frac{1}{a^2})}{N_1(\frac{1+P_2}{a^2})}$.
The first term \prettyref{eq:t1} can be bounded by applying \prettyref{cor:best} and \prettyref{eq:R2gap} with $B = a X_1$, $A = X_2$, $G= G_2$ and
$c=1$:
\begin{equation}
	h(a X_{1} + X_{2} + Z_{2}) - h(a X_{1}  + G_{2} + Z_{2}) \leq 
		n \sqrt{\frac{2 a^2 P_1 (C_2-R_2)\log e }{1+P_2}}.
	\label{eq:t3}
\end{equation}
Combining \prettyref{eq:t1} -- \prettyref{eq:t3} yields
\begin{equation}\label{eq:ct2}
	N_1\left({1\over a^2}\right) \le {\exp(2\delta)\over 1+P_2} N_1\left({1+P_2\over a^2}\right) \,.
\end{equation}
where $\delta$ is defined in \prettyref{eq:delta}.
From the concavity of $N_1(t)$ and~\eqref{eq:ct2}
\begin{align} N_1(1) &\le \gamma N_1\left({1\over a^2}\right) - (\gamma -1) N_1\left({1+P_2\over a^2}\right) \\
	  &\le N_1\left({1\over a^2}\right) \left(\gamma - (\gamma-1){1+P_2\over \exp(2\delta)}\right) , \label{eq:ct5}
\end{align}
where $\gamma = 1 + {1 - a^2\over P_2} > 1$. In view of \prettyref{eq:R1N1}, upper bounding $N_1\left({1/a^2}\right)$ in~\eqref{eq:ct5} via~\eqref{eq:ct3} we get after some simplifications the second part of~\eqref{eq:corner2}.

The outer bound for average power constraint \prettyref{eq:power2} follows analogously with \prettyref{eq:t3} replaced by \prettyref{eq:t4} below:
By \prettyref{prop:pqr}, the density of $a X_{1}  + G_{2} + Z_{2}$ is $(\frac{3\log e}{1+P2}, \frac{4a \log e \sqrt{nP_1}}{1+P2})$-regular. Applying \prettyref{prop:ppr} to \prettyref{eq:t3}, we have $h(a X_{1} + X_{2} + Z_{2}) - h(a X_{1}  + G_{2} + Z_{2}) \leq \Delta$, where
\[
\Delta =  (3 \sqrt{1+a^2 P_1+P_2} + 4 a \sqrt{P_1}) \frac{\log e}{1+P_2} \sqrt{n} W_2(P_{a X_{1} + X_{2} + Z_{2}} ,P_{a X_{1}  + G_{2} + Z_{2}}).
\]
Again using the fact that $W_2$ distance is non-decreasing under convolutions and invoking Talagrand's inequality, we have
$$ W_2(P_{a X_{1} + X_{2} + Z_{2}} ,P_{a X_{1}  + G_{2} + Z_{2}}) \leq W_2(P_{X_{2}+Z_2}, P_{G_{2}+Z_2}) \leq \sqrt{{2 (1+P_2)\over \log e} D(
P_{X_{2}+Z_2} \| P_{G_{2}+Z_2})}\,, $$ 
which yields
\begin{equation}
	h(a X_{1} + X_{2} + Z_{2}) - h(a X_{1}  + G_{2} + Z_{2}) \leq 
		n \sqrt{\frac{2  (C_2-R_2)\log e }{1+P_2}}  (3 \sqrt{1+a^2 P_1+P_2} + 4 a \sqrt{P_1}) .
	\label{eq:t4}
\end{equation}	
This yields the outer bound with $\delta'$ defined in \prettyref{eq:delta-avg}.

Finally, in both cases, when $R_2 \to C_2$, we have $\delta \to 0$ and $A \to \frac{1}{a^2}+\frac{P_2}{1+P_1}$ and hence from
\prettyref{eq:corner2} $R_1 \leq C_1'$. 
\end{proof}

\begin{remark} 
The first part of the bound~\eqref{eq:corner2} coincides with Sato's outer bound \cite{Sato78} and \cite[Theorem 2]{Kramer04} by Kramer, which \cite[Theorem 2]{Kramer04} was obtained by reducing the Z-interference channel to the degraded broadcast channel; the second part of \eqref{eq:corner2} 
is new, which settles the missing corner point of the capacity region (see \prettyref{sec:corner} for discussions).	Note that our estimates on $N_1(1)$ in the proof of Theorem~\ref{thm:corner2} are tight in the sense that there exists a concave
function $N_1(t)$ satisfying the listed general properties, estimates~\eqref{eq:ct2} and \eqref{eq:ct3} as well as attaining
the minimum of~\eqref{eq:ct4} and~\eqref{eq:ct5} at $N_1(1)$. Hence, tightening the bound via this method would
require inferring more information about $N_1(t)$.
\end{remark}

\begin{remark}
The outer bound \prettyref{eq:corner2} relies on Costa's EPI. To establish the second statement about corner point, it is sufficient to invoke the concavity of $\gamma \mapsto I(X_2; \sqrt{\gamma} X_2 + Z_2)$ \cite[Corollary 1]{guo.immse}, which is strictly weaker than Costa's EPI.
	\label{rmk:Igamma}
\end{remark}

The outer bound \prettyref{eq:corner2} is evaluated on Fig.~\ref{fig:eval} for the case of $b=0$ (Z-interference), where we also plot (just for reference) the simple Han-Kobayashi inner bound for the Z-GIC \prettyref{eq:ZGIC} attained by choosing $X_1 = U+V$ with $U\indep V$ jointly Gaussian. This achieves rates:
\begin{align}
\begin{cases}
	R_1 = {1\over 2} \log (1+P_1 - s) + {1\over 2} \log \left(1+ {a^2 s\over 1+a^2 (P_1-s) + P_2}\right)\\
  	 R_2 = {1\over 2} \log \left(1+ {P_2\over 1+a^2(P_1-s)}\right)
\end{cases}
\quad 0\le s\le P_1.
\end{align}   
For more sophisticated Han-Kobayashi bounds see~\cite{IS04_GIC,costa2011noisebergs}.
\begin{figure}
\centering
\includegraphics[width=.5\textwidth]{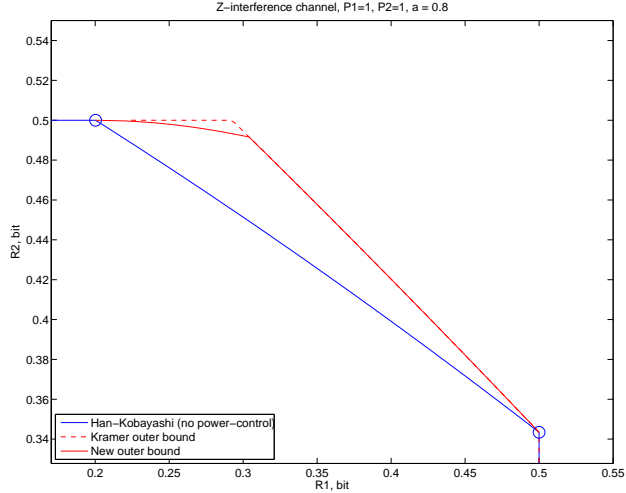}
\caption{Illustration of the ``missing corner point'': The bound in Theorem~\ref{thm:corner2} establishes the location
of the upper corner point, as conjectured by Costa~\cite{MC85_GIC}. The bottom corner point has been established by Sato \cite{Sato78}.}
\label{fig:eval}
\end{figure}

\subsection{Corner points of the capacity region}
\label{sec:corner}

The two corner points of the capacity region are defined as follows:
\begin{align} C_1'(a,b) &\eqdef \max\{R_1: (R_1,C_2) \in\matr(a,b)\}\,,\\
   C_2'(a,b) &\eqdef \max\{R_2: (C_1,R_2) \in\matr(a,b)\}\,,
\end{align} 
where $C_i = {1\over2}\log(1+P_i)$. As a corollary,
\prettyref{thm:corner2} completes the picture of the corner points for the capacity region of GIC for all
values of $a,b\in\mreals_+$ under the average power constraint \prettyref{eq:power2}. We note that the new result here is the proof of $C_1'(a,b) = \gacap{a^2 P_1\over 1+P_2}$ for $0 < a \leq 1$ and $b \geq 0$. The interpretation is that if one user desires to achieve its own interference-free capacity, then the other user must guarantee that its message is decodable at both receivers.
The achievability of this corner point was previously known, while the converse was previously considered by Costa~\cite{MC85_GIC} but with a
flawed proof, as pointed out in~\cite{IS04_GIC}. The high-level difference between our proof and that
of~\cite{MC85_GIC} is the replacement of Pinsker's inequality by Talagrand's and the use of a coupling argument.%
\footnote{After circulating our initial draft, we were informed that authors of~\cite{Bustin-ISIT14} posted an updated
manuscript~\cite{bustin2014effect} that also proves Costa's conjecture. Their method is based on the analysis of the
minimum mean-square error (MMSE) properties of good channel codes, but we were not able to verify all the details. A further update is in~\cite{bustin2015optimal}.}

Below we present a brief account of the corner points in various cases; for an extensive discussion see~\cite{IS13-GIC-corner}.
We start with a few simple observations about the capacity region $\matr(a,b)$:
\begin{itemize} 
\item 
Any rate pair satisfying the following belongs to $\matr(a,b)$:
\begin{equation}
	\begin{aligned} 
	R_1 &\le {1\over2} \log(1+P_1 \min(1,a^2))\,\\
	   R_2 &\le {1\over2} \log(1+P_2 \min(1,b^2))\,\\
	   R_1+R_2 &\le {1\over 2} \log (1+ \min(P_1 + b^2 P_2, P_2 + a^2 P_1)) ,
\end{aligned}	   
	\label{eq:mac-mac}
\end{equation}
which corresponds to the intersection of two Gaussian multiple-access (MAC) capacity regions, namely, $(X_1,X_2) \to Y_1$ and $(X_1,X_2) \to Y_2$. 
These rate pairs correspond to the case when each receiver decodes both messages. 

\item For $a>1$ and $b>1$ (strong interference) the capacity region is known to coincide with \prettyref{eq:mac-mac} \cite{Carleial75,Sato81}. So, without loss of generality we assume $a\le 1$ henceforth.
\item Replacing either $a$ or $b$ with zero can only enlarge the region (genie argument).
\item If $b\ge 1$ then for any $(R_1,R_2)\in\matr(a,b)$ we have \cite{Sato81}
	\begin{equation}\label{eq:bt1}
		R_1 + R_2 \le \gacap{b^2 P_2 + P_1}\,.
	\end{equation}	
	This follows from the observation that in this case $I(X_1, X_2; Y_1) = H(X_1,X_2)-o(n)$, since conditioned on $X_1$, $Y_2$ is a noisier
	observation of $X_2$ than $Y_1$.
\apxonly{\item Similarly, if $a\ge 1$ then for any $(R_1,R_2)\in\matr(a,b)$ we have
	\begin{equation}\label{eq:bt11}
		R_1 + R_2 \le \gacap{a^2 P_1 + P_2}\,.
	\end{equation}
}

\end{itemize}

For the top corner, we have the following:
\begin{equation}\label{eq:bt2}
	C_1'(a,b) = \begin{cases}
	\gacap{a^2 P_1\over 1+P_2}, & 0<a\le1, b \ge 0\\
	C_1, & a=0, b=0 \mbox{~or~} b \ge \sqrt{1+P_1}\\
	\gacap{P_1 + (b^2 -1) P_2\over 1+P_2}, & a=0, 1 < b < \sqrt{1+P_1}\\
	\gacap{P_1\over 1+b^2 P_2}, & a=0, 0<b \le 1.
	\end{cases} 
\end{equation}	
Note that for any $b \geq 0$, $a \mapsto C_1'(a,b)$ is discontinuous as $a\downarrow 0$. To verify~\eqref{eq:bt2} we consider each case separately:
\begin{enumerate}
\item For $a>0$ the converse bound follows from \prettyref{thm:corner2}. For achievability, we consider two cases. When
$b\le 1$, we have ${a^2P_1\over 1+P_2}\le \frac{P_1}{1+b^2 P_2}$ and therefore treating interference $X_2$ as noise at the first
receiver and using a Gaussian MAC-code for $(X_1,X_2)\to Y_2$ works.
For $b > 1$, the achievability follows from the MAC inner bound \prettyref{eq:mac-mac}. Note that since $ \gacap{P_1 + b^2 P_2} \ge \gacap{P_2 + a^2 P_1} $, a Gaussian MAC-code that works for $(X_1,X_2)\to Y_2$ will also work for $(X_1,X_2)\to Y_1$.
Alternatively, the achievability also follows from Han-Kobayashi inner bound (see, \eg, \cite[Theorem 6.4]{nit-book} with $(U_1,U_2)=(X_1,X_2)$ for $b \geq 1$ and $(U_1,U_2)=(X_1,0)$ for $b \leq 1$).
\item For $a=0$ and $b \ge \sqrt{1+P_1}$ the converse is obvious, while for achievability we have that ${b^2P_2\over 1+P_1}\le P_2$ and therefore $X_2$ is decodable at $Y_1$.
\item For $a=0$ and $1 < b < \sqrt{1+P_1}$ the converse is~\eqref{eq:bt1} and the achievability is just the MAC code
$(X_1,X_2)\to Y_1$ with rate $R_2=C_2$.
\item For $a=0$ and $0<b\le 1$ the result follows from the treatment of $C_2'(a,b)$ below by interchanging $a\leftrightarrow
b$ and $P_1\leftrightarrow P_2$.
\end{enumerate}

The bottom corner point is given by the following:
\begin{equation}
C_2'(a,b) = \begin{cases}
		\gacap{P_2\over 1+a^2 P_1}, &0\le a\le 1, b=0\mbox{~or~}b\ge\sqrt{1+P_1\over 1+a^2 P_1}\\
		\gacap{b^2 P_2\over 1+P_1}, &0\le a\le 1, 1 < b < \sqrt{1+P_1\over 1+a^2 P_1}\\
		\gacap{b^2 P_2\over 1+P_1}, &0\le a\le 1, 0<b\le 1
	\end{cases}	
	\label{eq:bcorner}
\end{equation}
which is discontinuous as $b\downarrow 0$ for any fixed $a \in [0,1]$.  We treat each case separately:
\begin{enumerate}
\item The case of $C_2'(a,0)$ is due to Sato \cite{Sato78} (see also \cite[Theorem 2]{Kramer04}). The converse part also follows
from~\prettyref{thm:corner2} (for $a=0$ there is nothing to prove). For the achievability, we notice that under $b\ge \sqrt{1+P_1\over 1+a^2 P_1}$ we have
$ {b^2 P_2 \over 1+P_1} > {P_2 \over 1+a^2 P_1} $
and thus $X_2$ at rate $C_2'(a,0)$ can be decoded and canceled from $Y_1$ by simply treating $X_1$ as
Gaussian noise (as usual, we assume Gaussian random codebooks). Thus the problem reduces to that of $b=0$. For $b=0$,
the Gaussian random coding achieves the claimed result if the second receiver treats $X_1$ as Gaussian noise. 
\item 
The converse follows from~\eqref{eq:bt1} and for the achievability we use the Gaussian MAC-code
$(X_1,X_2)\to Y_1$ and treat $X_1$ as Gaussian interference at $Y_2$. 
\item If $b\in(0,1]$, we apply results on $C_1'(a,b)$ in \prettyref{eq:bt2} by interchanging $a\leftrightarrow b$ and $P_1\leftrightarrow P_2$.
\end{enumerate}

\section{Discrete version}
\label{sec:discrete}

\subsection{Bounding entropy and information via Ornstein's distance}

Fix a finite alphabet $\matx$ and an integer $n$. On the product space $\matx^n$ we define the Hamming distance
$$ d_H(x,y) = \sum_{j=1}^n \indc{x_j \neq y_j}\,,$$
and consider the corresponding Wasserstein distance $W_1$. In fact, $\frac{1}{n} W_1(P,Q)$ is known as Ornstein's $\bar d$-distance \cite{GNS75,Marton86}, namely,
\begin{equation}
	\bar d(P,Q) = \frac{1}{n} \inf \Expect[d_H(X,Y)],
	\label{eq:dbar}
\end{equation} 
where the infimum is taken over all couplings $P_{XY}$ of $P$ and $Q$. 
For $n=1$, this coincides with the total variation, which is also expressible as $\TV(P,Q)=\frac{1}{2} \int |dP-dQ|$ for $P,Q$ on $\calX$.

For a pair of distributions $P,Q$ on $\matx^n$ we may ask  the following questions:
\begin{enumerate}
\item Does $D(P\|Q)$ control the entropy difference $H(P)-H(Q)$?
\item Does $\bar d(P,Q)$ control the entropy difference $H(P)-H(Q)$?
\end{enumerate}
Recall that in the Euclidean space the answer to both questions was negative unless the distributions satisfy certain
regularity conditions. For discrete alphabets the answer to the first question is still negative in general (see \prettyref{sec:intro} for a counterexample); nevertheless, the answer to the second one turns out to be positive:

\begin{proposition}\label{prop:dew} Let $P$ and $Q$ be distributions on $\matx^n$ and let
$$ F_{\matx}(x) \eqdef x \log (|\matx|-1) + x \log {1\over x} + (1-x) \log {1\over 1-x}\,, \quad  0\leq x \leq 1.$$
Then 
\begin{equation}\label{eq:jpt}
	|H(P) - H(Q)| \le n F_{\matx}(\bar d(P,Q))\,.
\end{equation}
\end{proposition}
\begin{proof} In fact, the statement holds for any translation-invariant distance $d(\cdot,\cdot)$ on $\calX$ extended
additively to $\calX^n$, \ie, $d(x,x')=\sum_{i=1}^n d(x_i,x'_i)$ for any $x,x'\in \calX^n$. Indeed, define
$$ f_n(s) \eqdef \max_{P_X} \left\{{1\over n}H(X): \EE[d(X, x_0)] \le n s \right\}\,,$$
where $x_0\in \matx^n$ is an arbitrary fixed string. It is easy to see that $s \mapsto f_n(s)$ is concave since
$P\mapsto H(P)$ is. Furthermore, writing $X=(X_1,\ldots,X_n)$ and applying chain-rule for entropy we get
$$ f_n(s) = f_1(s)\,.$$
Thus, letting $X,Y$ be distributed according to the $\bar d$-optimal coupling of $P$ and $Q$, we get
\begin{align} 
H(X) - H(Y) &\le H(X,Y)- H(Y) = H(X|Y) \\
		&\le n \EE\left[f_n(\EE[d(X,Y)|Y])\right] \label{eq:jpt1}\\
		&\le n f_n(\bar d(P,Q))\,,\label{eq:jpt2}
\end{align}
where~\eqref{eq:jpt1} is by definition of $f_n(\cdot)$ and~\eqref{eq:jpt2} is by Jensen's inequality. Finally, for the Hamming
distance we have $f_1(s) = F_{\matx}(s)$ by Fano's inequality. 
\end{proof}

Notice that the right-hand side of~\eqref{eq:jpt} behaves like $n \bar d \log \frac{1}{\bar d}$ when $\bar d(P,Q)$ is small. This super-linear dependence is in fact sharp.\footnote{To see this, consider $Q = \Bern(p)^{\otimes n}$ and choose $P$ to be the output distribution of the optimal lossy compressor for $Q$ at average distortion $\delta n$. By definition, $\bar d(P,Q) \leq \delta$. On the other hand,  $H(P) = n (h(p) - h(\delta)+o(1))$ as $n\diverge$ and hence $|H(P)-H(Q)| = n (h(\delta)+o(1))$, which asymptotically meets the upper bound \prettyref{eq:jpt} with equality.}
Nevertheless, if certain regularity of distributions is assumed, the estimate \prettyref{eq:jpt} can be improved to be
linear in $\bar d(P,Q)$. The next result is the analog of \prettyref{prop:ppr} in the discrete space. We formulate it in
a form convenient for applications in multi-user information theory. 

\begin{prop} 
\label{prop:dbar}
Let $P_{Y|X,A}$ be a two-input blocklength-$n$ memoryless channel, namely
\[
P_{Y|X,A}(y|x,a) = \prod_{j=1}^n W(y_j|x_j, a_j),
\]
 where $W(\cdot|\cdot)$ is a stochastic matrix and $y\in \maty^n,
x\in\matx^n, a\in\mata^n$.
Let $X,A,\tilde A$ be independent $n$-dimensional discrete random vectors. 
Let $Y$ and $\tY$ be the outputs generated by $(X, A)$ and $(X, \tA)$, respectively. 
Then
 \begin{align} |H(Y) - H(\tY)| &\le c n \bar d(P_{Y}, P_{\tY}) \label{eq:dbarH}\\
 	D(P_{Y} \| P_{\tY}) + D(P_{\tY} \| P_{Y}) &\le 2 c n \bar d(P_{Y}, P_{\tY}) \label{eq:dbarD}
	  \\
 	|I(X; Y) - I(X; \tY)| &\le 2 c n \Expect[\bar{d}(P_{Y|X},P_{\tY|X})] \label{eq:dbarI}
\end{align}
where
\begin{align} 
c & ~\eqdef \max_{x, a, y, y'} \log {W(y|x,a)\over W(y'|x,a)} , \label{eq:c}\\
\Expect[\bar{d}(P_{Y|X},P_{\tY|X})] & ~\eqdef \sum_{x\in \matx^n} P_X(x) \bar d(P_{Y|X=x}, P_{\tY|X=x}).
\end{align}
\end{prop}
\begin{proof} 
Given any stochastic matrix $G$, define $L(G)\triangleq \max_{u,u',v} \log \frac{G(u|v)}{G(u'|v)}$.
Recall the following fact from \cite[Eqn.~(58)]{PV12-optcodes} about mixtures of product distributions: Let $U$ and $V$ be $n$-dimensional discrete random vector connected by a product channel, that is, $P_{V|U}=\prod_{i=1}^n P_{V_i|U_i}$. Then the mapping $v\mapsto \log P_V(v)$ is $L$-Lipschitz with respect to the Hamming distance, where $L=\max_{j\in[n]} L(P_{V_i|U_i})$. 
Consider another pair $(\tU,\tV)$ connected by the same channel, \ie, $P_{\tV|\tU}=P_{V|U}$.
Then Lipschitz continuity implies that
$\EE\left[\left| \log \frac{P_{V}(V)}{P_{V}(\tV)} \right|\right] \le L \EE[d_H(V,\tV)]$ for any coupling $P_{V\tV}$. Optimizing over the coupling and in view of \prettyref{eq:dbar}, we obtain
$$ \EE\left[\left| \log \frac{P_{V}(V)}{P_{V}(\tV)} \right|\right] \le L n \bar d(P_{V}, P_{\tV})\,.$$
Repeating the proof of~\eqref{eq:ppr2}--\eqref{eq:ppr4}, we have
\begin{align}
	|H(V)-H(\tV)| \leq &~L n \bar{d}(P_V,P_\tV) \label{eq:Hstab} \\
 	D(P_{V} \| P_{\tV}) + D(P_{\tV} \| P_{V}) \leq &~ 2 c n \bar d(P_{V}, P_{\tV}) \label{eq:Dstab}	
\end{align}
Applying \prettyref{eq:Hstab} and \prettyref{eq:Dstab} to $Y$ and $(X,A)$ gives \prettyref{eq:dbarH} and \prettyref{eq:dbarD} with $L=c$ defined in \prettyref{eq:c}.

To bound the mutual information, we first notice
$$ |I(X;Y) - I(X; \tY)| \le |H(Y) - H(\tY)| + |H(Y|X) - H(\tY|X)|\,. $$
Applying~\eqref{eq:Hstab} conditioned on $X=x$ we get
\[
	|H(Y|X=x)-H(\tY|X=x)| \le c_x n \bar d(P_{Y|X=x}, P_{\tY|X=x})\,,
\]
where $c_x = \max_{j\in[n]} \max_{y,y',a} \log \frac{W(y|x_j,a)}{W(y'|x_j,a)}$. Note that $c_x \leq c$ for any $x$, averaging over $P_X$ gives
\begin{equation}\label{eq:rgf1}
	|H(Y|X)-H(\tY|X)| \le cn \Expect[\bar d(P_{Y|X}, P_{\tY|X})]\,.
\end{equation}
From the convexity of $(P,Q) \mapsto \bar d(P,Q)$, which holds for any Wasserstein distance, we have
$ \bar d(P_{Y}, P_{\tY}) \le \Expect[\bar d(P_{Y|X}, P_{\tY|X})]$
and so the left-hand side of~\eqref{eq:rgf1} also bounds $|H(Y)-H(\tY)|$ from above.
\end{proof}

\subsection{Marton's transportation inequality}

In this section we discuss how previous bounds (\prettyref{prop:dew} and \ref{prop:dbar}) in terms of the $\bar{d}$-distance can be converted to bounds in terms of KL divergence. This is possible when $Q$
is a product distribution, thanks to Marton's transportation inequality~\cite[Lemma 1]{Marton86}. We formulate this together with a few 
other properties of the $\bard$-distance in the following lemma proved in \prettyref{app:dbar}. 
\begin{lemma}
\label{lmm:dbar}~~~~~~~~
	\begin{enumerate}
	\item (Marton's transportation inequality~\cite{Marton86}): For any pair of distributions $P$ and $Q=\prod_{i=1}^n Q_i$ on
	$\matx^n$,
	\begin{equation}
		\bar d\left(P,\, Q\right) \leq \sqrt{\frac{D(P \| Q)}{2 n\log e}}.
		\label{eq:marton1}
	\end{equation}
	\item (Tensorization) $\bar{d}(\prod_{i=1}^n P_i, \prod_{i=1}^n Q_i) \leq \frac{1}{n} \sum_{i=1}^n \TV(P_i,Q_i)$.
	
	\item (Contraction) For $P_{XY}$ and $Q_{XY}$ such that $P_{Y|X}=Q_{Y|X}=\prod_{i=1}^n P_{Y_i|X_i}$, 
	\begin{equation}
	\bar{d}(P_Y,Q_Y) \leq \max_{i\in [n]} \etaTV(P_{Y_i|X_i})   \bar{d}(P_X,Q_X) .
		\label{eq:dbar-contract}
	\end{equation}
	where $\etaTV(W)$ is Dobrushin's contraction coefficient of a Markov kernel $W$ defined as $\etaTV(W)=\sup_{x,x'} \TV(W(\cdot|x), W(\cdot|x'))$.
\end{enumerate}
\end{lemma}

If we assume that $D(P \| Q) = \epsilon n$ for some small $\epsilon$, then combining \prettyref{eq:jpt} and \prettyref{eq:marton1} gives
\[
	|H(P) - H(Q)| \le n F_{\matx}\left( \sqrt{{D(P \| Q) \over2 n \log e} } \right)\,,
\]
where the right-hand side behaves as $n \sqrt{\epsilon} \log \frac{1}{\epsilon}$ when $\epsilon \to 0$. This estimate has a
one-sided improvement (here again $Q$ must be a product distribution):
\begin{equation}
H(P) - H(Q) \le \sqrt{\frac{2 n D(P \| Q)}{\log e}} \log |\calX| 
	\label{eq:chang}
\end{equation}
(see~\cite{cchang} for $n=1$ and~\cite[Appendix H]{WV10} for the general case).
\apxonly{

Strangely, with methods of this paper we can obtain an incomparable to~\eqref{eq:chang} one-sided estimate. Again, when $Q=Q_1^{\otimes n}$ we have
$$ H(P)-H(Q) \le c W_1(P,Q) - D(P\|Q) \le  
		c \sqrt{{n\over2\log e} D(P\|Q)} - D(P\|Q)\,,$$
where $c = \max_{a_i\in \supp Q}\log {Q_1(a_1)\over Q_1(a_2)}$.
}
\apxonly{

Another observation: Although~\prettyref{prop:dew} allows one to prove $H(P)\approx H(Q)$ when $\bar d(P,Q)\ll
1$ without any regularity assumptions, it is not possible to prove $D(P\|Q)\ll n$ without such assumptions. Indeed, just
taking $P$ to be uniform on some linear subspace of $\FF_2^n$ and $Q$ to be uniform on the translate of the subspace
provides a counter-example.
}

Switching to the setting in \prettyref{prop:dbar}, let us consider the case where $\tA$ has
i.i.d.~components, \ie, $P_{\tA}=P_0^{\otimes n}$. 
Define
\begin{align} 
   \etaTV &\eqdef \max_{x, a,  a'} \TV(W(\cdot|x,a), W(\cdot|x, a'))\,,
\end{align}
which is the maximal Dobrushin contraction coefficients among all channels
$W(\cdot|\cdot, x)$ indexed by $x \in \calX$.
Then
\begin{equation}
	\bar d(P_Y, P_\tY) \le 	\Expect[\bar d(P_{Y|X}, P_{\tY|X})] \leq
	\etaTV \bar d(P_A, P_\tA) \le \etaTV \sqrt{\frac{D(P_{A} \| P_{\tA})}{2 n \log e}}\,,
	\label{eq:dbar-etatv}
\end{equation}
where the left inequality is by convexity of the $\bar d$-distance as a Wasserstein distance, the middle inequality is by \prettyref{lmm:dbar},
and the right inequality is via~\eqref{eq:marton1}.
An alternative to the estimate \prettyref{eq:dbar-etatv} is the following:
\begin{align} 
\bar d(P_Y, P_\tY) &\le 	\Expect[\bar d(P_{Y|X}, P_{\tY|X})] \label{eq:gb1}\\
		   &\le \Expect\qth{ \sqrt{{1 \over2n\log e} D(P_{Y|X} \| P_{\tY|X})}}
		   	\label{eq:gb2}\\
		   &\le \sqrt{{1\over2n\log e} D(P_{Y|X} \| P_{\tY|X}|P_X)}
		   	\label{eq:gb3}\\
		   &\le \sqrt{{1\over2n\log e} \etaKL D(P_{A} \| P_{\tA})}
		   	\label{eq:gb4}
\end{align}
where~\eqref{eq:gb2} is by~\eqref{eq:marton1} since $P_{\tY|X=x}$ is a product distribution as $\tA$ has a product distribution,~\eqref{eq:gb3} is
by Jensen's inequality, and~\eqref{eq:gb4} is by the tensorization property of the strong data-processing constant for divergence~\cite{AG76}:
$$ \etaKL \eqdef \max_{x, Q_0} {D(\sum_a Q_0(a) W(\cdot|x,a)   \| \sum_a P_0(a) W(\cdot |x, a))\over
					D(P_0\|Q_0)} \,.$$
To conclude, in the regime of $D(P_A\|P_\tA) \leq \epsilon n$ for some small $\epsilon$ our main \prettyref{prop:dbar} yields
\begin{equation}
 |H(Y) - H(\tY)| \simleq n \sqrt{\epsilon} 
	\label{eq:HY-twosided}
\end{equation}
matching the behavior of~\eqref{eq:chang}. However, the estimate \prettyref{eq:HY-twosided} is stronger, because a) it is two-sided and b) $P_\tY$ can be a
mixture of product distributions (since $X$ in \prettyref{prop:dbar} may be arbitrary).

\apxonly{
\textbf{Remarks:}
\begin{enumerate}
\item Notice that the actual sequential Pinsker produces coupling:
$$ W_1(P_{\tilde X}, P_{G}) \le \sqrt{1\over 2\log e} \sum_j \EE[\sqrt{D(P_{\tilde X_j|\tilde
X^{j-1}=X^{j-1}}|P_G)}]\,. $$
\item Also for one-directional estimate $H(Y)-H(Y_G)\le \cdots$ we can further improve the constant $c$ by replacing
it with the one for the channel $\sum_a P_{Y|X=x,A=a} P_G(a)$.
\end{enumerate}
}

\subsection{Application: corner points for discrete interference channels}

In order to apply~\prettyref{prop:dbar} to determine corner points of capacity regions of discrete memoryless interference channels (DMIC) we
will need an auxiliary tensorization result. This result appears to be a rather standard exercise for degraded channels
and so we defer the proof to \prettyref{app:tensoh}.

\begin{prop}\label{prop:tensoh} 
Given channels $P_{A|X}$ and $P_{B|A}$ on finite alphabets, define
\begin{equation}\label{eq:tensoh}
	F_c(t) \eqdef \max\{H(X|A,U)\colon H(X|B,U)\le t, U\to X \to A \to B\}\,.
\end{equation}
Then the following hold:
\begin{enumerate}
	\item (Property of $F_c$) The function $F_c:\mreals_+\to\mreals_+$ is concave, non-decreasing and $F_c(t)\le t$. Furthermore, $F_c(t)<t$ for all $t>0$, provided that $P_{B|A}$ and $P_{A|X}$ satisfy
	\begin{equation}
	P_{B|A=a} \not\perp P_{B|A=a'}, \quad \forall a \neq a'
	\label{eq:strict1}
\end{equation}
	 and 
	 \begin{equation}
	P_{A|X=x}\neq P_{A|X=x'}, \quad \forall x\neq x',
	\label{eq:strict2}
\end{equation}
respectively.

\item (Tensorization) For any blocklength-$n$ Markov chain
$X^n\to A^n \to B^n$, where $P_{A^n|X^n} = P_{A|X}^{\otimes n}$ and $P_{B^n|A^n}=P_{B|A}^{\otimes n}$ are $n$-letter memoryless channels, we have
\begin{equation}
H(X^n|A^n) \le n F_c\left({1\over n} H(X^n|B^n)\right)\,.	
	\label{eq:sl}
\end{equation}
\end{enumerate}
\end{prop}

\begin{remark}
	Neither of the sufficient condition  \prettyref{eq:strict1} and \prettyref{eq:strict2} for strict inequality is superfluous, as can be seen from the example $B=A$ and $A \indep X$, respectively; in both cases $F_c(t)=t$.
\end{remark}

\apxonly{
Example of $F_c$ in BSC: The region without $U$ is a convex curve parameterized by $X\sim \Bern(p)$. And the region with $U$ is simply the epigraph of the curve.
}

The important consequence of~\prettyref{prop:tensoh} is the following implication:\footnote{This is the analog of the
following property of Gaussian channels, exploited in~\prettyref{thm:corner2} in the form of Costa's EPI: For
i.i.d.~Gaussian $Z$ and $t_1<t_2<t_3$ we have
$$I(X; X+t_2 Z) = I(X; X+t_3 Z) + o(n) \implies I(X; X+ t_1 Z) =I(X; X+t_3 Z) + o(n)\,.$$
This also follows from the concavity of $\gamma \mapsto I(X; \sqrt{\gamma} X + Z)$.}
\begin{coro}
	\label{cor:tensoX}
Let $X^n\to A^n\to B^n$, where the memoryless channels $P_{A|X}$ and $P_{B|A}$ of blocklength $n$ satisfy the conditions \prettyref{eq:strict1} and \prettyref{eq:strict2}.
	Then there exists a continuous function $g: \reals_+ \to \reals_+$ satisfying $g(0)=0$, such that for all $n$
	\begin{equation}\label{eq:tensoX}
	I(X^n;A^n) \leq I(X^n;B^n) + \epsilon n \quad\implies\quad H(X^n) \leq I(X^n;B^n) + g(\epsilon) n\,,
\end{equation}
\end{coro}
\begin{proof}
By \prettyref{prop:tensoh}, we have $F_c(t)<t$ for all $t>0$. This together with the concavity of $F_c$ implies that $t \mapsto t - F_c(t)$ is convex, strictly increasing and strictly positive on $(0,\infty)$. Define $g$ as the inverse of $t \mapsto t - F_c(t)$, which is increasing and concave and satisfies $g(0)=0$.
Since $I(X^n;A^n) \leq I(X^n;B^n) + \epsilon n$, the tensorization result \prettyref{eq:tensoh} yields
$$H(X^n|B^n)\le H(X^n|A^n) + \epsilon n \leq n F_c\pth{\frac{1}{n} H(X^n|B^n)} + \epsilon n,$$
\ie, $ t \leq F_c(t) + \epsilon$, where $t\eqdef {1\over n} H(X^n|B^n) $. Then $t \leq g(\epsilon)$ by definition, completing the proof.
\end{proof}

We are now ready to state a non-trivial example of corner points for the capacity region of DMIC.
The proof strategy mirrors that of \prettyref{thm:corner2}, with \prettyref{cor:best} and Costa's EPI replaced by
\prettyref{prop:dbar} and \prettyref{cor:tensoX}, respectively. 
\begin{theorem}
\label{thm:corner-discrete} 
Consider the two-user DMIC:
\begin{align} Y_1 &= X_1\,,\\
   Y_2 &= X_2 + X_1 + Z_2 \mod 3\,,
\end{align}
where $X_1 \in \{0,1,2\}^n$, $X_2\in \{0,1\}^n, Z_2 \in \{0,1,2\}^n$ are independent and $Z_2 \sim P_2^{\otimes n}$ is i.i.d.~for some non-uniform $P_2$ containing no zeros. The maximal rate
achievable by user 2 is 
\begin{equation}
C_2 = \max_{\supp(Q) \subset \{0,1\}} H(Q * P_2) - H(P_2).
	\label{eq:C2-discrete}
\end{equation}
At this rate the maximal rate of user 1 is 
\begin{equation}
C_1' = \log 3 - \max_{\supp(Q) \subset \{0,1\}} H(Q * P_2).
	\label{eq:C1p-discrete}
\end{equation}
\end{theorem}

\begin{remark}
	As an example, consider $P_2=\left[{1-\delta}, {\delta\over2}, {\delta\over 2}\right]$ where $\delta \neq 0,1,\frac{1}{3}$. Then the maximum in \prettyref{eq:C2-discrete} is achieved by $Q=[\frac{1}{2},\frac{1}{2}]$. Therefore $C_2=H(P_3) - H(P_2)$ and $C_1'=\log 3 - H(P_3)$, where $P_3=\left[{2-\delta\over 4}, {2-\delta\over4}, {\delta\over 2}\right]$. Note that in the case of $\delta=\frac{1}{3}$, where \prettyref{thm:corner-discrete} is not applicable, we simply have $C_2 = 0$ and $C_1'=\log 2$ since $X_2 \indep Y_2$. Therefore the corner point is discontinuous in $\delta$.
\end{remark}

\begin{remark}
	\prettyref{thm:corner-discrete} continues to hold even if cost constraints are imposed. Indeed, if $X_2 \in \{0,1,2\}^n$ is required to satisfy 
	\[
	\sum_{i=1}^n \sfb(X_{2,i}) \leq n B
	\]
	for some cost function $\sfb: \{0,1,2\} \to \reals$. Then the maximum in \prettyref{eq:C2-discrete} and
	\prettyref{eq:C1p-discrete} is taken over all $Q$ such that $\Expect_Q[\sfb(U)] \leq B$. Note that taking
	$B=\infty$ is equivalent to dropping the constraint $X_2\in\{0,1\}^n$ in \prettyref{eq:C2-discrete}. In this
	case, $C_1'=0$ which can be shown by a simpler argument not involving \prettyref{prop:dbar}.
\end{remark}

\begin{proof}
\apxonly{
Steps:
\begin{enumerate}
\item $\max I(X_2; Y_2|X_1) = nC_2$ achieved by uniform $X_2$ because $Z_2 \stackrel{d}{=}-Z_2$.
\item Under such distribution on $X_2$ we have $(X_2+Z_2) \eqdef Z_3 \sim P_3^{\otimes n}$.
\item Achievability: The rate pair $(C_1', C_2)$ is achievable by an MAC-code for $X_1X_2 \to Y_2$ with both $X_1$ and $X_2$ uniform on their respective alphabet.
\item Proceed to converse. Usual argument shows
	$$ D(P_{X_2 + Z_2}\| P_{Z_3}) = o(n) $$
\item Marton
	$$ \bar d(P_{X_2 + Z_2}, P_{Z_3})=o(1) $$
\item Convolution reduces Wasserstein distance:
	$$ \bar d(P_{X_1+X_2 + Z_2}, P_{X_1+Z_3})=o(1) $$
\item \prettyref{prop:dbar}
	$$ I(X_1; Y_2) = I(X_1; X_1 + Z_3) + o(n)\,.$$
\item Usual stuff
	$$ I(X_1; Y_2 | X_2) = I(X_1; Y_2) + o(n) $$
\item From~\eqref{eq:tensoX}
	$$ nR_1 = H(X_1) \le I(X_1; X_1 + Z_3) + o(n) \le nC_1' + o(n) $$
\end{enumerate}
}

We start with the converse.
 Given a sequence of codes with vanishing probability of error and rate pairs $(R_1,R_2)$, where $R_2=C_2 - \epsilon$, we show that $R_1 \leq C_1'-\epsilon'$, where $\epsilon' \to 0$ as $\epsilon \to 0$.
Let $Q_2$ be the maximizer of \prettyref{eq:C2-discrete}, \ie, the capacity-achieving distribution of the channel $X_2 \mapsto X_2+Z_2$.
Let $\tilde X_2 \in \{0,1\}^n$ be distributed according to $Q_2^n$. Then $\tilde X_2+Z_2  \sim P_3^{\otimes n}$, where $P_3 = Q_2 * P_2$.
By Fano's inequality,
\begin{align}
n(C_2 - \epsilon + o(1)) = n (R_2+o(1)) = & ~ I(X_2;Y_2) 	\nonumber \\
\leq & ~ I(X_2; Y_2|X_1) = I(X_2; X_2+Z_2)	\\
= & ~ n C_2 - D(P_{X_2 + Z_2}\| P_{\tX_2+Z_2}),
\end{align}
that is,
	$$ D(P_{X_2 + Z_2}\| P_{\tX_2+Z_2}) \leq n \epsilon + o(n). $$
Since $P_{\tX_2+Z_2} = P_3^{\otimes n}$ is a product distribution, 
Marton's inequality \prettyref{eq:marton1} yields
	$$ \bar d(P_{X_1+X_2 + Z_2}, P_{X_1+\tX_2+Z_2}) \leq \bar d(P_{X_2 + Z_2}, P_{\tX_2+Z_2}) \leq \sqrt{\frac{\epsilon}{2 \log e}} + o(1).  $$
Applying \prettyref{eq:dbarI} in \prettyref{prop:dbar} and in view of the translation invariance of the $\bar{d}$-distance, we obtain
	\begin{align}
	|I(X_1; Y_2) - I(X_1; X_1 + \tX_2+Z_2)| 
= & ~ |I(X_1; X_1+X_2+Z_2) - I(X_1; X_1+\tX_2 + Z_2)| 	\nonumber \\
\leq & ~ 2 c n \Expect[\bar{d}(P_{X_1+X_2+Z_2|X_1}, P_{X_1+\tX_2 + Z_2|X_1}) ]	\nonumber \\
= & ~ 2 c n \bar{d}(P_{X_2+Z_2}, P_{\tX_2 + Z_2}) 	\nonumber \\
\leq & ~ 	(\alpha \sqrt{\epsilon} + o(1)) n, \nonumber 
\end{align}
	where $c = \max_{z,z'\in\{0,1,2\}} \log \frac{P_2(z)}{P_2(z')}$ and
	$\alpha = \frac{2c}{\sqrt{2\log e}} $ are finite since $P_2$ contains no zeros by assumption. 
On the other hand, 
	$$ I(X_1; X_1+Z_2) = I(X_1; Y_2 | X_2) = I(X_1; Y_2) + I(X_1;X_2|Y_2) = I(X_1; Y_2) + o(n), $$
	where $I(X_1;X_2|Y_2) \leq H(X_2|Y_2) = o(n)$ by Fano's inequality. Combining the last two displays, we have
	\[
	I(X_1; X_1+\tX_2 + Z_2) \leq  I(X_1; X_1+Z_2) + (\alpha \sqrt{\epsilon} + o(1)) n.
	\]
 Next we apply \prettyref{cor:tensoX}, with $X=X_1 \to A=X_1+Z_2 \to B=A+\tX_2$. To verify the conditions, note that the channel $P_{A|X}$ is memoryless and additive with non-uniform noise distribution $P_2$, which satisfies the condition \prettyref{eq:strict2}. Similar, the channel $P_{B|A}$ is memoryless and additive with noise distribution $Q_2$ , which is the maximizer of \prettyref{eq:C2-discrete}. Since $P_2$ is not uniform, $Q_2$ is not a point mass. Therefore $P_{B|A}$ satisfies \prettyref{eq:strict1}. Then \prettyref{cor:tensoX} yields 
	$$ nR_1 = H(X_1) \le I(X_1; X_1 + \tX_2+Z_2) + g(\alpha \sqrt{\epsilon}) n \le nC_1' + o(n), $$
	where the last inequality follows from the fact that $\max_{X_1} I(X_1; X_1 + \tX_2+Z_2) = n C_1'$ attained by $X_1$ uniform on $\{0,1,2\}^n$.
	
Finally, note that the rate pair $(C_1', C_2)$ is achievable by a random MAC-code for $(X_1,X_2) \to Y_2$, with $X_1$ uniform
on $\{0,1,2\}^n$ and $X_2 \sim Q_2^{\otimes n}$.
\apxonly{
Details: the MAC capacity region contains
\begin{align}
R_1 \leq & ~ I(X_1;Y_2|X_2) = I(X_1;X_1+Z_2) = \log 3 - H(P_2)	\\
R_2 \leq & ~ I(X_2;Y_2|X_1) = I(X_2;X_2+Z_2) = C_2 = H(Q_2*P_2) - H(P_2) \\
R_1+R_2 \leq & ~ I(X_1X_2;Y_2) = I(X_1+X_2;X_1+X_2+Z_2) = \log 3 - H(P_2),
\end{align}
which is satisfied by $(C_1',C_2)$.}
\end{proof}

\apxonly{
\subsection{Other GIC stuff}
As an aside, we provide a converse which recovers the known corner point due to Sato \cite{Sato78}:
\begin{theorem} Suppose $a\le 1$. Then the capacity region of the GIC satisfies:
\begin{equation}
R_{2} \le  C_2' + \frac{1+P_1}{1/a^{2}+P_1}  (C_1-R_1) \,. 	
	\label{eq:sato}
\end{equation}
 where $C_1 = {1\over 2} \log (1+P_{1})$ and $C_2' = {1\over 2} \log(1+{P_{2}\over a^2 P_{1} +1})$.
\label{thm:sato}
\end{theorem}

\begin{remark}
	The converse in \prettyref{thm:sato} is superseded by Kramer's outer bound \cite[Theorem 2]{Kramer04}: $R_1 \leq \frac{1}{2}\log(1+P_1')$ and $R_2 \leq \frac{1}{2}\log(1+\frac{a^2 P_2'}{1+a^2 P_1'})$ for some $P_1',P_2'>0$ such that $P_1'+P_2'=P_1+P_2/a^2$. Setting $\alpha = \frac{1+P_1}{1/a^{2}+P_1} \in (0,1)$, this implies $\alpha R_1+R_2 \leq \sup_{P_1'+P_2'=P_1+P_2/a^2}  [\frac{1}{2}\log(1+P_1') + \frac{1}{2}\log(1+\frac{a^2 P_2'}{1+a^2 P_1'})] = \alpha C_1+C_2'$ achieved at $P_1'=P_1$, which is exactly \prettyref{eq:sato}.
	
	\label{rmk:kramer}
\end{remark}

\begin{proof}
Similarly to the previous proof, consider
\begin{align} 
R_{1} = {1\over n} I(X_1; Y_1) & \leq C_1 -{1\over n} D(X_1 + Z_1 \| G_1 + Z_1)\,,
\end{align}	
where $G_1 \sim \matn(0, P_{1} I_n)$. We use the strong data processing inequality of divergence with Gaussian
base measure~\cite[Theorem 6]{EC98}: 
For any zero-mean $X$ with variance $t$,
$$D(P_X * \calN(0,s) \| \calN(0,s+t)) \leq \frac{t}{t+s} D(P_X\|\calN(0,s))\,.$$
This is also what the standard estimate~\eqref{eq:ptt1} yields.

	Since $a\le 1$, we have
\begin{align}
\frac{1}{n} D(aX_1 + Z_1 \| aG_1 + Z_1)
= & ~ \frac{1}{n} D(X_1 + Z_1 + W \| G_1 + Z_1 + W)	\nonumber \\
\leq  & ~ \frac{1+P_1}{1/a^{2}+P_1} \frac{1}{n} D(X_1 + Z_1 \| G_1 + Z_1)\\
	\leq  & ~ \frac{1+P_1}{1/a^{2}+P_1} (C_1-R_1) \label{eq:bdd1}
\end{align}
where $W \sim \calN(0,(1/a^2-1)I_n)$.

Next, we have\footnote{If $\EE[|X_1|^2] = nP_{1}$, then \prettyref{eq:bdd2} holds with equality.} 
\begin{align} I(X_{2}; Y_{2}) &=h(X_{2} + a X_{1} + Z_{2}) - h(a X_{1} + Z_{2}) \nonumber \\
		& \leq h(X_{2} + a X_{1} + Z_{2}) - h(a G_{1} + Z_{2}) + D(a G_{1} + Z_{2} \| a X_{1} + Z_{2})\label{eq:bdd2}  \\
		& \leq n C_2' + \frac{1+P_1}{1/a^{2}+P_1} n (C_1-R_1) \label{eq:bdd3}
\end{align}
where \prettyref{eq:bdd3} follows from $\EE[|X_1|^2] \leq nP_{1}$ and \prettyref{eq:bdd1}.
\end{proof}
}

\apxonly{
\appendices
\section{$W_1$-regularity of smoothed differential entropy}
	\label{sec:hw1}
This appendix collects a few other regularity results for differential entropy of Gaussian convolutions in the scalar case. The downside is that these results do not  tensorize to $n$-dimensional spaces.

\begin{theorem}
Let $X_1,X_2$ be such that $\Expect[X_1^2],\Expect[X_2^2] \leq \sigma^2$. Let $Z\sim\calN(0,\sigma_Z^2)$. Then
\begin{equation}
|h(X_1+Z) - h(X_2+Z)| \leq  \kappa \pth{\frac{\sigma^2+\sigma_Z^2}{\sigma_Z^4}}^{1/5}  W_1^{2/5}(P_{X_1},P_{X_2}),
	\label{eq:hw1}
\end{equation}
where $\kappa$ is some universal constant.
	\label{thm:hw1}
\end{theorem}
\begin{remark}
Observations:
\begin{enumerate}
	\item Note that both sides are invariant if $X_1,X_2,Z$ are scaled simultaneously. 
	\item In the special case of $X_2\sim \calN(0,\sigma^2)$ and $\Expect[X_1^2]=\sigma^2$, we have a different bound by coupling
	\[
	|h(X_1+Z) - h(X_2+Z)| = D(P_{X_1+Z}\|P_{X_2+Z}) \leq \frac{1}{2 \sigma_Z^2} W_2^{2}(P_{X_1},P_{X_2}).	
	\]
	which appears to be incomparable to \prettyref{eq:hw1}.
\end{enumerate}

\end{remark}

\begin{proof}
Let $Y_i=X_i+Z$ with density $f_i$. Denote the pdf of $Z$ by $p_{Z}$. Then $\Lip(p_Z) = \frac{1}{\sigma^2_Z} \Lip(\varphi) = \frac{1}{\sqrt{2 e \pi } \sigma^2_Z} \triangleq L$. Then we have
\begin{equation}
\|f_1 - f_2\|_\infty = \sup_y |\Expect[f_Z(y-X_1) - f_Z(y-X_2)]| \leq L W_1(P_{X_1},P_{X_2}),	
	\label{eq:linf}
\end{equation}
 by applying the optimal coupling that gives $\Expect[|X_1-X_2|]=W_1(P_{X_1},P_{X_2})$, and
	\begin{equation}
	\Lip(f_i) \leq \Lip(p_Z) = L.
	\label{eq:lip}
\end{equation} 

Fix $c>0$.	Let $g(t) = \log \frac{1}{t \vee c}$. Then $\log \frac{1}{t} = g(t) + \log \frac{c}{t} \indc{t<c}$. Then
	\begin{align}
	& ~ |h(Y_1) - h(Y_2)|	\nonumber \\
= & ~ |\Expect[\log f_1(Y_1)] - \Expect[\log f_2(Y_2)]|	\nonumber \\
\leq & ~ |\Expect[g(f_1(Y_1)) - g(f_2(Y_2))]| + \Big|\Expect\Big[\log \frac{f_1(Y_1)}{c} \indc{f_1(Y_1) \leq c}\Big]\Big| + \Big|\Expect\Big[\log \frac{f_2(Y_2)}{c} \indc{f_2(Y_2) \leq c}\Big]\Big| \label{eq:b2}
\end{align}
For the first term, note that $\Lip(g) = \frac{1}{c}$. Furthermore, 
\begin{equation}
|f_1(Y_1) - f_2(Y_2)| \leq|f_1(Y_1) - f_2(Y_1)| + |f_2(Y_1) - f_2(Y_2)| \leq \|f_1 - f_2 \|_\infty + \Lip(f_2)|Y_1 - Y_2|.
	\label{eq:fY12}
\end{equation}
Optimizing over the coupling of $Y_1,Y_2$, we have
\begin{equation}
|\Expect[g(f_1(Y_1)) - g(f_2(Y_2))]| \leq \Lip(g) (\|f_1 - f_2 \|_\infty + \Lip(f_2) W_1(Y_1, Y_2)) \leq \frac{2L}{c} W_1(X_1,X_2).	
	\label{eq:a}
\end{equation}

To bound the remaining terms, we show that
\begin{equation}
\sup \sth{\int_\reals p \log \frac{c}{p} \indc{p < c} \diff y\colon \int_\reals y^2 p \diff y \leq s, p \geq 0} \leq	\kappa' s^{1/3} c^{2/3}.
	\label{eq:opt}
\end{equation}
for some universal constant $\kappa'$. 
Note that
\begin{align}
\int_\reals p \log \frac{c}{p} \indc{p < c} \diff y &= \int_\reals \int_0^\infty \log \frac{c}{e t} \diff t \indc{t <
p(y) < c} \diff y = \int_0^c \log \frac{c}{e t}  \cdot m\{y: t < p(y) < c\} \diff t\\
	&\le \int_0^{c/e} \log \frac{c}{e t}  \cdot m\{y: t < p(y) < c\} \diff t
	\label{eq:yp}
\end{align}
Put $E_t = \{y: t < p(y) < c\}$. Then $$s \geq \int_\reals y^2 p \diff y \geq t \int_{E_t} y^2 \diff y \geq t
\int_{-m(E_t)/2}^{m(E_t)/2} y^2 \diff y = \frac{t}{12} m(E_t)^3,$$ where the last inequality follows from the symmetric
rearrangement inequality \cite[Theorem 3.4]{lieb.loss}. Therefore $m(E_t) \leq (\frac{12s}{t})^{1/3}$. Plugging into \prettyref{eq:yp}, we have
\[
\int_\reals p \log \frac{c}{p} \indc{p < c} \diff y \leq s^{1/3} c^{2/3} \cdot \left(12\over e^2\right)^{1/3} \frac{9}{4}\,.
\]
In view of \prettyref{eq:a}, applying \prettyref{eq:opt} with $s=\sigma^2+\sigma_Z^2$ to the second and third terms in \prettyref{eq:b2} and optimizing over $c$, we have
\[
|h(Y_1) - h(Y_2)| \leq \inf_{c>0} \sth{\frac{2L}{c} W_1(P_{X_1},P_{X_2}) + \kappa' s^{1/3} c^{2/3}} =  \kappa  s^{1/5} (L W_1(P_{X_1},P_{X_2}))^{2/5}.\qedhere
\]
\end{proof}

\subsection{Extensions}
A variant of \prettyref{thm:hw1} is the following result (still scale-invariant) involving total variation. The proof follows the same argument using the Hamming distance in lieu of the $L_1$-distance.
\begin{theorem}
Under the same assumption of \prettyref{thm:hw1}, for another universal constant $\kappa$,
\begin{equation}
|h(X_1+Z) - h(X_2+Z)| \leq  \kappa \pth{1+\frac{\sigma^2}{\sigma_Z^2}}^{1/5}  \TV^{2/5}(P_{X_1},P_{X_2}).	
	\label{eq:htv}
\end{equation}
	\label{thm:htv}
\end{theorem}
\begin{proof}
The proof follows the same argument in that of \prettyref{thm:hw1}, except that \prettyref{eq:linf} and \prettyref{eq:fY12} are respectively replaced by
	\[
\|f_1 - f_2\|_\infty = \sup_y |\Expect[f_Z(y-X_1) - f_Z(y-X_2)]| \leq \|p_Z\|_\infty \TV(P_{X_1},P_{X_2}),	
	\]
	and
\[
|f_1(Y_1) - f_2(Y_2)| \leq|f_1(Y_1) - f_2(Y_1)| + |f_2(Y_1) - f_2(Y_2)| \leq \|f_1 - f_2 \|_\infty + \|p_Z\|_\infty \indc{Y_1 \neq Y_2},
	\]
	where we have used $\|f_i\|_\infty \leq \|p_Z\|_\infty = \frac{1}{\sqrt{2 \pi} \sigma_Z}$.
\end{proof}

Note that the proof of \prettyref{thm:hw1} does not crucially rely on the Gaussianity on $Z$. It only uses the simple fact that $\Lip(p * \mu) \leq \Lip(p)$ for any distribution $\mu$. In particular, we can replace $Z$ by any Gaussian mixture. In particular, we have
\begin{theorem}
Let $(X,\hat X),B,Z$ be independent such that $\Expect[X^2],\Expect[\hat{X}^2] \leq \sigma^2$, $\Expect[B^2] \leq \sigma_B^2$ and $Z\sim\calN(0,\sigma_Z^2)$. Then
\begin{equation}
|h(X+B+Z) - h(\hat{X}+B+Z)| \leq  \kappa \pth{\frac{\sigma^2+\sigma_B^2+\sigma_Z^2}{\sigma_Z^4}}^{1/5}  W_1^{2/5}(P_{X},P_{\hat{X}}).	
	\label{eq:hw1b}
\end{equation}
	\label{thm:hw1b}
\end{theorem}

\section{A standalone account}
Define the following $n$-dimensional random vectors:
\begin{align*}
Y = & ~ B + X + Z	\nonumber \\
\tilde Y = & ~ B + \tilde X + Z
\end{align*}
where $B,X,\tilde X,Z$ are independent, $Z \sim N(0,I_n)$, $\tilde X$ is Gaussian with the same mean $\mu_X$ as $X$ and covariance matrix $s I_n$.
\begin{theorem}
Assume that $\|B\| \leq \sqrt{n \beta}$ almost surely and $\Expect[\|X\|^2] \leq \Expect[\|\tX\|^2]$.
Then
\begin{equation}
	h(Y) - h(\tilde Y) \leq \sqrt{\frac{2\beta n}{1+s} D(X+Z \| \tilde X+ Z)}.
	\label{eq:ch}
\end{equation}	
	\label{thm:ch}
\end{theorem}
\begin{proof}
Denote the density of $Y,\tilde Y$ and $Z$ by $p, \tilde p$ and $\varphi$, respectively. Note that 
\begin{align}
h(Y) - h(\tY)
= & ~ \int p \log \frac{1}{p} - \int \tp \log \frac{1}{\tp}	\nonumber \\
= & ~ \int (p-\tp) \log \frac{1}{\tp} - \int p \log \frac{p}{\tp} 	\nonumber \\
= & ~ \int (p-\tp) \log \frac{\varphi_{1+s}}{\tp} + \frac{\Expect \|Y\|^2 - \Expect \|\tY\|^2}{2(1+s)}   - D(p\|\tp) 	\nonumber \\
\leq & ~ \int (p-\tp) \log \frac{\varphi_{1+s}}{\tp}   \label{eq:hd}
\end{align}
where we used the fact that $\Expect \|Y\|^2 - \Expect \|\tY\|^2 = \Expect \|X\|^2 - \Expect \|\tX\|^2 \leq 0$.
Note that $\tilde p$ equals the convolution of $P_B$ and $\calN(0,(1+s) I_n)$. Then
$\nabla \log \frac{\varphi_{1+s}}{\tilde p}(y)  = - \frac{1}{1+s} \Expect[B|\tY=y]$ and $\|B\| \leq \sqrt{n\beta}$ a.s.,
we have $\|\nabla \log \frac{\varphi_{1+s}}{\tp}(y)\| \leq \frac{\sqrt{n\beta} }{1+s} \triangleq L$ for almost all, and hence all, $y \in \reals$ as $\log \frac{\varphi_{1+s}}{\tp}$ is smooth; in other words, $\log \frac{\varphi_{1+s}}{\tp}: \reals\to\reals$ is $L$-Lipschitz. Therefore
\begin{equation}
\left|\int (p-\tp) \log \frac{\varphi_{1+s}}{\tp} \right| \leq L W_1(Y,\tY).	
	\label{eq:ch3}
\end{equation}
Next we bound $W_1(Y,\tY)$ by applying Talagrand's inequality \cite{Talagrand96}: Note that $\tX + Z$ is Gaussian with covariance matrix $(1+s)I_n$. Then
\begin{align}
W_1(Y,\tY)
\leq & ~ W_1(X+Z,\tX+Z)	\label{eq:w1} \\
\leq & ~ W_2(X+Z,\tX+Z)	\label{eq:w2} \\
\leq & ~ \sqrt{2(1+s) D(X+Z\|\tX+Z)}	\label{eq:w3} 
\end{align}
Assembling \prettyref{eq:hd}, \prettyref{eq:ch3} and  \prettyref{eq:w3} yields the desired \prettyref{eq:ch}.
\end{proof}
}

\section*{Acknowledgment}
Explaining that the ``missing corner point'' requires proving of~\eqref{eq:costa_req}, as well as the majority of our
knowledge on interference channels were provided by Prof.~Chandra Nair. We acknowledge his scholarship and
patience deeply.

The research of Y.P. has been supported in part by the Center for Science of Information (CSoI),
an NSF Science and Technology Center, under grant agreement CCF-09-39370 and by the NSF CAREER award under grant
agreement CCF-12-53205.  
The research of Y.W. has been supported in part by NSF grants IIS-14-47879, CCF-14-23088 and CCF-15-27105.
This work would not be possible without the generous support of the Simons Institute for the Theory
of Computing and California SB-420.

\appendix
\section{Proof of \prettyref{lmm:dbar}}
	\label{app:dbar}
\begin{proof}
To prove the tensorization inequality, let $(X,Y)=(X_i,Y_i)_{i=1}^n$ be independent and individually distributed as the optimal coupling of $(P_i,Q_i)$. Then 
$\Expect[d_H(X,Y)]=\sum_{i=1}^n \prob{X_i\neq Y_i} = \sum_{i=1}^n \TV(P_i,Q_i)$.

To show \prettyref{eq:dbar-contract}, 
let $\pi_{X,Y,\tX,\tY}$ be an arbitrary coupling of $P_{XY}$ and $Q_{XY}$ so that $(X,\tX)$ is distributed according to the optimal coupling of $\dbar(P_X,Q_X)$, that is, $\Expect_{\pi}[d_H(X,\tX)]=n\dbar(P_X,Q_X)$.
By the first inequality we just proved, for any $x,x'\in\calX^n$, 
\[
\dbar(P_{Y|X=x}, P_{Y|X=\tx}) \leq \frac{1}{n} \sum_{i=1}^n \TV(P_{Y_i|X_i=x_i}, P_{Y_i|X_i=\tx_i}) \leq \frac{1}{n} \sum_{i=1}^n \etaTV(P_{Y_i|X_i}) \indc{x_i\neq \tx_i} \leq \frac{\eta d_H(x,\tx)}{n} .
\]
where $\eta = \max_{i\in[n]} \etaTV(P_{Y_i|X_i})$ and the middle inequality follows from Dobrushin's contraction coefficient.
Applying Dobrushin's contractoin \cite{RLD70} (see \cite[Proposition 18]{PW14a}, with $\rho=\frac{1}{n} d_H$ and $r=\eta\rho$), there exists a coupling $\pi'_{X,Y,\tX,\tY}$ of $P_{XY}$ and $Q_{XY}$, so that $\pi'_{X\tX}=\pi_{X\tX}$ and $\Expect_{\pi'}[d_H(Y,\tY)] \leq \eta \Expect_{\pi}[d_H(X,\tX)] = n \eta \dbar(P_X,Q_X)$, concluding the proof.
\end{proof}

\section{Proof of \prettyref{prop:tensoh}}
\label{app:tensoh}
\begin{proof} Basic properties of $F_c$ follow from standard arguments. To show the strict inequality $F_c(t)<t$ under the conditions \prettyref{eq:strict1} and \prettyref{eq:strict2}, we first
notice that $F_c$ is simply the concave envelope of the
set of achievable pairs $(H(X|A),H(X|B))$ obtained by iterating over all $P_X$. By Caratheodory's theorem, it is sufficient to consider a ternary-valued
$U$ in the optimization defining $F_c(t)$. Then the set of achievable pairs $(H(X|A,U), H(X|B,U))$ is convex
and compact (as the continuous image of the compact set of distributions $P_{U,X}$). Consequently, to have $F_c(t)=t$
there must exist a distribution $P_{U,X}$, such that
\begin{equation}\label{eq:th0}
	H(X|A,U) = H(X|B,U)=t\,.
\end{equation}
We next show that under the extra conditions on $P_{B|A}$ and $P_{A|X}$ we must have $t=0$. Indeed, 
\prettyref{eq:strict1} guarantees the channel $P_{B|A}$ satisfies the strong data processing inequality  (see, \eg, \cite[Exercise 15.12 (b)]{ckbook2} and \cite[Section 1.2]{PW14a} for a survey)
that there exists~$\eta<1$ such that 
\begin{equation}\label{eq:th1}
	I(X;B|U)\le \eta I(X; A|U).
\end{equation}
 From~\eqref{eq:th0} and~\eqref{eq:th1} we infer
that $I(X;A|U)=0$, or equivalently
$$ D(P_{A|X}\|P_{A|U} | P_{U,X})=0\,.$$
On the other hand, the condition \prettyref{eq:strict2} ensures that then we must have $H(X|U)=0$. Clearly, this implies $t=0$
in~\eqref{eq:th0}.

To show the single-letterization statement \eqref{eq:sl}, we only consider the case of $n=2$ since the
generalization is straightforward by induction. Let 
$ X^2\to A^2\to B^2 $
be a Markov chain with blocklength-$2$ memoryless channel in between. We have
\begin{align} H(X^2|B^2) &= H(X_1|B^2) + H(X_2|B^2,X_1)\\
	&= H(X_1|B^2) + H(X_2|B_2, X_1)\label{eq:th2}\\
	&\ge H(X_1|B_1,A_2) + H(X_2|B_2,X_1)\label{eq:th3} 
\end{align}
where~\eqref{eq:th2} is because $B_2 \to X_2 \to X_1 \to B_1$ and hence $I(X_2;B_1|X_2B_2)=0$, 
and~\eqref{eq:th3} is because $B_1\to X_1\to A_2 \to B_2$. Next consider the chain
\begin{align} H(X | A^2) &= H(X_1|A^2) + H(X_2|A^2,X_1)\\
	        &= H(X_1|A^2) + H(X_2 | A_2,X_1)\label{eq:th4}\\
		&\leq F_c(H(X_1|B_1,A_2)) + F_c(H(X_2|B_2,X_1))\label{eq:th5}\\
		&\le 2F_c\left({1\over2}H(X_1|B_1,A_2) + {1\over2}H(X_2|B_2,X_1)\right)\label{eq:th6}\\
		&\le 2F_c\left({1\over2}H(X^2|B^2)\right)\label{eq:th7}
\end{align}
where~\eqref{eq:th4} is by $A_2\to X_2 \to X_1 \to A_1$ and hence $I(X_2;A_1|X_1,A_2)=0$,~\eqref{eq:th5} is by the definition of $F_c$ and since we have both
$ A_2\to X_1\to A_1\to B_1$ and $X_1\to X_2\to A_2\to B_2$,
\eqref{eq:th6} is by the concavity of $F_c$, and finally~\eqref{eq:th7} is by the monotonicity of $F_c$ and \eqref{eq:th3}.
\end{proof}


\end{document}